\documentclass[aps,showpacs,amsmath,amssymb]{revtex4}
\usepackage{amsmath}
\usepackage{latexsym}
\usepackage{amsfonts}
\usepackage{amssymb}
\usepackage{amsthm}
\usepackage{color}
\usepackage{bbm,dsfont}
\usepackage{graphicx}
\usepackage{hyperref}
\usepackage{enumerate}

\newtheorem{proposition}{Proposition}
\newtheorem{proposition?}{Proposition?}
\newtheorem{theorem}{Theorem}
\newtheorem{lemma}{Lemma}
\newtheorem{corollary}{Corollary}

\newtheorem{example}{Example}
\newtheorem{condition}{Condition}

\newcommand{\hi}{\mathcal{H}} 
\newcommand{\hik}{\mathcal{K}} 
\newcommand{\id}{\mathbf{1}} 
\newcommand{\nul}{0} 

\newcommand{\A}{\mathsf{A}}
\newcommand{\B}{\mathsf{B}}
\newcommand{\E}{\mathsf{E}}
\renewcommand{\P}{\mathsf{P}}

  \makeatletter
  \def\mathcomposite{%
     \@ifstar
        {\def\@mathcomposite@option{%
            \baselineskip\z@skip\lineskiplimit-\maxdimen}%
         \@mathcomposite}%
        {\let\@mathcomposite@option\offinterlineskip
         \@mathcomposite}}
  \def\@mathcomposite{%
     \@ifnextchar[\@@mathcomposite{\@@mathcomposite[0]}}
  \def\@@mathcomposite[#1]#2#3#4{%
     #2{\mathchoice
        {\@mathcomposite@{#1}{#3}{#4}\displaystyle{1}}%
        {\@mathcomposite@{#1}{#3}{#4}\textstyle{1}}%
        {\@mathcomposite@{#1}{#3}{#4}%
         \scriptstyle\defaultscriptratio}%
        {\@mathcomposite@{#1}{#3}{#4}%
         \scriptscriptstyle\defaultscriptscriptratio}}}
  \def\@mathcomposite@#1#2#3#4#5{%
     \vcenter{\m@th\@mathcomposite@option
        \dimen@\f@size\p@\dimen@#1\dimen@\dimen@#5\dimen@
        \divide\dimen@ 18
        \edef\@mathcomposite@skipamount{\the\dimen@}%
        \ialign{\hfil$#4##$\hfil\cr
           #2\crcr
           \noalign{\vskip\@mathcomposite@skipamount}%
           #3\crcr}}}
  \makeatother

\begin{document}

\title[Energy-Time Uncertainty Relations in Quantum Measurements]{Energy-Time Uncertainty Relations in Quantum Measurements}

\author{Takayuki Miyadera}
\email{miyadera@nucleng.kyoto-u.ac.jp}
\affiliation{Department of Nuclear Engineering, Kyoto University - 6068501
Kyoto, Japan}


\begin{abstract}
Quantum measurement is a physical process. A system and an apparatus 
interact for a certain time period (measurement time), and
during this interaction, information about an observable is 
transferred from the system to the apparatus. 
In this study, we quantify the energy fluctuation of the quantum apparatus 
 required 
for this physical process to occur autonomously. 
We first examine the so-called standard model of measurement, 
which is free from any non-trivial energy-time uncertainty relation,   
to find that it needs an external system that switches on the interaction
between the system and the apparatus. 
In such a sense this model is not closed. Therefore to 
treat a measurement process in a fully quantum manner we need to consider 
a ``larger" 
quantum apparatus which works also as a timing device switching on the 
interaction.  
%
In this setting we prove that a trade-off relation (energy-time 
uncertainty relation), $\tau\cdot \Delta H_A
\geq \frac{\pi \hbar}{4}$, holds between the energy fluctuation 
$\Delta H_A$ of 
the quantum apparatus and the measurement time $\tau$. 
We use this trade-off relation 
to discuss the spacetime uncertainty relation concerning  
the operational meaning of the microscopic structure of spacetime. 
In addition, we derive another trade-off inequality between 
the measurement time and 
the strength of interaction between the system 
and the apparatus. 
\end{abstract}

\maketitle

\section{Introduction}\label{section1}
Quantum measurement is one of the simplest, yet most important, 
physical processes \cite{WZbook, vNbook, BLMbook, BGLbook, HZbook, Bell}. 
Without measurement, we see no event and 
 obtain no information. A quantum measurement
is a process that maps the quantum state of a quantum system 
to the classical state (probability distribution) of an external classical system
that belongs to the observer side. 
It is known 
that the interface (border)
between quantum and classical systems  
can be shifted. 
In fact, an observer does not interact with the system itself. 
Instead, they extract information from another system, 
a measurement apparatus, 
that has direct contact with the system. 
Both the system and the measurement apparatus can be treated 
quantum mechanically.
%
\par
The main purpose of this paper is to investigate 
the energy and interaction required for measuring an observable. 
More precisely, we investigate the energy of a quantum apparatus 
and the strength of interaction between a system and the apparatus
so that the process is fully described by quantum theory.  
To put the question from a more pragmatic viewpoint, 
our interest is in the ``amount of resource" required to perform 
a measurement (or information transfer) task 
\cite{StateReduction}. 
Thus, to study this problem, the interface
between quantum and classical systems
must be located such that the apparatus is treated 
quantum mechanically. 
For instance, although we have an equivalent 
minimum description (called an instrument \cite{Oza1984,
BLMbook, BGLbook, HZbook}) 
of
the dynamics  
and measurement results 
that puts the interface between the system and the apparatus, 
our problem does not allow us to employ it because our interest is in 
the limitations of the quantum apparatus.  
In section \ref{section2}, 
we discuss how ``large" the quantum side (or quantum apparatus) must be 
to describe a measurement process in a fully quantum manner.
We examine a so-called standard measurement model 
to find that the model requires an external system switching on the
measurement interaction between a system and an apparatus.
In this sense, the measurement process is not closed - it is not described 
in a fully quantum manner in this model. 
This model does not obey any non-trivial energy-time uncertainty relation.  
This conclusion agrees with previous results \cite{Aharonov, Busch2, BuschTEUR}.
Then, in section \ref{section3}, we rigorously formulate 
a quantum apparatus as a timing device that 
autonomously switches on a perturbation. 
An argument on the analyticity shows that the total Hamiltonian must be 
two-side unbounded to fulfill the conditions for the timing device. 
Our main results
are presented in section \ref{section4}. 
We consider a measurement apparatus that 
autonomously switches on the interaction 
at a certain time. 
We show that 
there is an energy-time uncertainty type 
relation between the fluctuation of the apparatus Hamiltonian 
and time duration of the measurement. 
The proof of this trade-off relation is given by combining 
two kinds of uncertainty relations. 
We first observe that 
a perfect measurement of a given observable implies a perfect 
distortion of the conjugate states. This property is called an 
information-disturbance relation and is related to an uncertainty 
relation for a joint measurement. We then 
employ the Mandelstam-Tamm uncertainty relation 
to obtain a restriction on time and energy required for the distortion.  
In addition, 
because the measurement process involves an information transfer 
from the system to the apparatus, the interaction between 
them should not be too weak. 
In fact, in section \ref{section5} we show 
a trade-off relation between the strength of interaction and 
the measurement time. 
The Robertson uncertainty relation plays an essential role
in the proof.
In section \ref{section6} we illustrate an argument of 
the spacetime uncertainty relation as an application of our result. 
\section{Measurement process as a physical process}\label{section2}
Let us consider a quantum system described by a Hilbert 
space $\hi$. A quantum state is represented as a density operator 
on $\hi$. We first present a description when we put an interface  
between the quantum and external classical systems just outside this 
system. 
By measuring an observable, we obtain an outcome. 
In general, the outcome is not definite and only a probability distribution on 
the outcome set is determined. 
Thus the measurement of an observable 
maps the quantum state of a quantum system 
to the classical state (probability distribution) of an external classical system
which belongs to the observer side. 
For a map to be consistent with an interpretation based on 
probability, it must be an affine map. 
Each affine map is completely described by a positive-operator-valued measure (POVM) 
in the quantum system (see e.g., \cite{BGLbook}). 
A POVM (with a discrete outcome set $\Omega$) is defined as a 
family of positive operators $\{\A(x)\}_{x\in \Omega_{\A}}$ 
satisfying $\sum_{x\in \Omega_{\A}}\A(x) =\id$. 
A POVM $\A=\{\A(x)\}_{\Omega_{\A}}$ 
gives a probability distribution $P_{\A}(x):= \mbox{tr}[\rho \A(x)]$ on 
an outcome set $\Omega_{\A}$. 
A POVM $\A=\{\A(x)\}_{x\in \Omega_{\A}}$ is called a projection-valued measure (PVM), or a sharp observable, if 
each $\A(x)$ is a projection operator. 
\par
Next we shift the interface between the quantum and classical systems. 
The measurement process is a 
physical interaction between a system 
and a measurement 
apparatus. 
We treat the apparatus as a quantum system 
described by a Hilbert space $\hik$. By taking $\hik$ to be sufficiently large, 
we can expect the total system to be (approximately) closed. 
Thus the total dynamics is specified by a Hamiltonian $H$. 
The total Hamiltonian can be decomposed into three parts 
\cite{Hamiltonian}: 
\begin{eqnarray*}
H=H_S + H_A +V, 
\end{eqnarray*}
where $H_S$ (resp. $H_A$) acts only on $\hi$ (resp. $\hik$) and 
$V$ describes an interaction Hamiltonian acting on $\hi \otimes \hik$. 
(More precisely, the total Hamiltonian should be written as 
$H= H_S \otimes \id + \id \otimes H_A +V$.) 
The initial state of the composite system at time $t=t_0$ is described 
by $\Theta(t_0):=\rho\otimes \sigma$ with an unknown state $\rho$ of the 
system and 
a fixed initial state $\sigma$ of the apparatus. This evolves according to 
the von Neumann equation up to a certain predetermined time $t=t_0+\tau$,  
where the time duration $\tau$ is called a measurement time duration. 
Then an observer outside these quantum systems observes a 
meter observable $\E=\{\E_x\}$, which is a POVM on the apparatus. 
This whole process is said to describe a measurement process of 
an observable $\A=\{\A_x\}$ if 
$\mbox{tr}[\A_x \rho] = \mbox{tr}[(\id \otimes \E_x)e^{-i\frac{H \tau}{\hbar}}
(\rho \otimes \sigma) e^{- \frac{H\tau}{\hbar}}]$ holds for every $\rho$. 
\par
The above physical measurement model 
is consistent with the abstract 
measurement theory by Ozawa's extension theorem \cite{Oza1984} and has been shown 
to be useful in analyzing real measurement processes. 
This model, however, is not sufficient for obtaining general results on the 
minimum energy fluctuation required for a measurement. In order to illustrate
this point, we consider the following example, a so-called standard 
measurement model developed essentially by von Neumann \cite{vNbook}. 
Suppose that a system is a qubit (thus $\hi = \mathbb{C}^2$) and 
an apparatus is a one-particle system on a real line (see FIG. 
\ref{figure_switch11}).
\begin{figure}
\includegraphics[width=9cm]{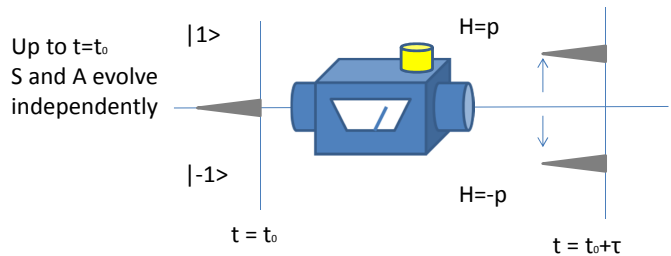}
\caption{A measurement model with $H_A=0$.}
\label{figure_switch11} 
\end{figure}
Set $H_S=\nul$, $H_A=\nul$ and $V=\sigma_z \otimes p$, 
where $p$ represents the momentum operator.
As the meter observable we employ $\{\int_{x\geq 0} dx |x\rangle \langle x|, 
\int_{x <0} dx |x\rangle \langle x|\}$.
For an arbitrary time duration $\tau$, 
if the initial state at time $t=t_0$ of the apparatus is prepared 
in a narrowly-localized state with respect to a position, 
the accuracy of the measurement of $\sigma_z$ 
can be made arbitrarily high. Therefore, there is no 
limitation on the time and energy required for this measurement process.
\par
By looking carefully at this model, one may notice that the
 time duration is given by hand. 
In this model, the interaction between a system and an apparatus is switched on 
at $t=t_0$ and switched off at $t=t_0+\tau$. The dynamics outside this time interval is 
not discussed. An external observer must 
put these systems together at $t=t_0$, which 
until then must have been somehow independent. 
Thus the mechanism initiating 
this switching-on process is supposed to be outside the quantum model.
In this sense, this model is not closed and we need to shift the interface further
to include the on-and-off switching process.
\par
In the discussion below, we develop a general measurement process that also treats  
this switching-on (and off) mechanism quantum mechanically.
In some situations, roughly speaking, we shift the quantum-classical interface so that the quantum side includes 
the experimenter who controls the apparatus. 
In the next section we give a mathematically rigorous formulation 
of the timing device that switches on a perturbation. 
\section{Apparatus as a switching device for an interaction}\label{section3}
As stated above, we investigate an apparatus 
that works not only as 
an information extractor but also as a switching device for an interaction (perturbation). 
To do so, in this section, we study a formulation of the latter condition and 
derive some results.  
We consider an apparatus that works 
as a timing device to switch on an interaction 
at a certain time after $t=t_0$. 
Thus, we assume that at time $t=t_0$ the state of the total 
system is a product state.  
The apparatus is specified by the Hilbert space $\hik$, 
Hamiltonian $H_A$ and a state $\sigma(t_0)$ at time $t=t_0$. 
\par
Let us first 
assume that the apparatus 
is perfectly isolated and has no interaction with the system. 
Then the time evolution is described by a von Neumann equation 
determined by an apparatus Hamiltonian $H_A$ which acts only on 
$\hik$. 
For an arbitrary time $t \in \mathbb{R}$, the state at time $t$ is 
written as 
$\sigma(t) = e^{-i \frac{H_A (t-t_0)}{\hbar}}\sigma(t_0) 
e^{i \frac{H_A (t-t_0)}{\hbar}}$. 
\par
Next, we consider an interaction between 
the system and apparatus. 
We denote the Hilbert space 
of the system
 by $\hi$. The individual dynamics of the system is governed by the 
Hamiltonian $H_S$ defined on $\hi$. 
Thus the total Hamiltonian has three parts, 
$H =H_S + H_A +V$, where $V$ is the interaction 
Hamiltonian that acts on the tensor product $\hi \otimes \hik$.  
The time evolution of a state $\Theta(t)$ (a density operator) over the composite system 
is described by the corresponding von Neumann equation, 
\begin{eqnarray*}
i \hbar \frac{d}{dt} \Theta(t) = [H, \Theta(t)]. 
\end{eqnarray*}
\par
Let us consider conditions for $H_A$, $V$ and a state of the apparatus 
to describe a timing device. 
They must satisfy the
following general conditions
that represent the capability to switch on the interaction  
at a certain time. 
The apparatus 
evolves freely 
up to time $t=t_0$ and  
then switches 
on an interaction with the system at some time after $t=t_0$. 
Thus, the state of the apparatus should be decoupled from that of the system 
until time $t=t_0$. 
The time evolution of the state is determined entirely 
once the state at a certain time is given. We impose the following 
condition on the dynamics and the state at time $t=t_0$: 
\begin{condition}\label{cond1} 
(No interaction up to time $t=t_0$) 
An apparatus (specified by $H_A$, $V$ and a state $\sigma(t_0)$ 
at time $t=t_0$) satisfies {\em the no interaction condition up to time $t=t_0$}
if for any $t\leq t_0$ and for any state $\rho$ of 
the system, 
\begin{eqnarray}
[V, \rho \otimes \sigma(t) ] =0
\label{cond1eq}
\end{eqnarray}
holds, where $\sigma(t):= e^{-i \frac{H_A (t-t_0)}{\hbar}}
\sigma(t_0) e^{i \frac{H_A (t-t_0)}{\hbar}}$.   
\end{condition}
We emphasize that in Condition \ref{cond1}, the left-hand side of (\ref{cond1eq}) 
must be vanishing for any  
$\rho$. This condition implies 
that only after the 
apparatus reaches $\sigma(t_0)$, does $V$ induce the apparatus and the system 
to interact. In fact the following lemma shows that 
these conditions are equivalent.
\begin{lemma}\label{lemmaequiv}
Condition \ref{cond1} is satisfied if and only if 
for any 
state $\rho$ of the system and $t\leq t_0$, 
\begin{eqnarray}
e^{-i \frac{H(t-t_0)}{\hbar}}(\rho \otimes \sigma(t_0) )
e^{i \frac{H(t-t_0)}{\hbar}}
= \rho(t) \otimes \sigma(t)
\label{cond1int} 
\end{eqnarray}
holds, where $\rho(t)$ is defined by
$\rho(t) = e^{-i \frac{H_S (t-t_0)}{\hbar}}
\rho e^{i \frac{H_S (t-t_0)}{\hbar}}.$ 
\end{lemma}
\begin{proof}
``If part": 
We differentiate by $t$, 
\begin{eqnarray*}
e^{-i \frac{H(t-t_0)}{\hbar}}(\rho \otimes \sigma(t_0) )
e^{i \frac{H(t-t_0)}{\hbar}}
= \rho(t) \otimes \sigma(t)
\end{eqnarray*}
to obtain 
\begin{eqnarray*}
[H, e^{-i \frac{H(t-t_0)}{\hbar} }
(\rho \otimes \sigma(t_0)) e^{i \frac{H(t-t_0)}{\hbar}}] 
= [H_S + H_A, \rho(t) \otimes \sigma(t)].
\end{eqnarray*}
As the left-hand side is written as 
$[H_S +H_A +V, \rho(t)\otimes \sigma(t)]$,
we obtain 
$[V, \rho(t)\otimes \sigma(t)]=0$. 
As $\rho$ is arbitrary, 
we obtain Condition \ref{cond1}. 
\\
``Only if" part: 
It holds that for $t\leq t_0$, 
\begin{eqnarray*}
i \hbar \frac{d}{dt} 
\rho(t) \otimes \sigma(t)
&=&[ H_S +H_A, \rho(t)\otimes \sigma(t)]
\\
&=& [H, \rho(t)\otimes \sigma(t)], 
\end{eqnarray*}
where we used (\ref{cond1eq}) for $\rho(t)$.
Thus we conclude that $\rho(t)\otimes \sigma(t)$ 
satisfies the von Neumann equation whose integration is 
(\ref{cond1int}). 
\end{proof}
This condition ensures that 
if once the state at an arbitrary time $t_1<t_0$ is prepared 
as $\rho(t_1) \otimes \sigma(t_1)$, 
the composite system evolves independently 
up to time $t=t_0$. 
\begin{lemma}\label{lemma1}
If Condition \ref{cond1} is satisfied and 
at a certain time $t_1 \leq t_0$ the state has a product form
$\Theta(t_1)=\rho(t_1)\otimes \sigma(t_1)$; then  
the state $\Theta(t)$ at any time $t\leq t_0$ becomes the product state; 
\begin{eqnarray*}
\Theta(t)= \rho(t) 
\otimes \sigma(t),  
\end{eqnarray*}
where 
$\rho(t)=
e^{-i \frac{H_S}{\hbar}(t-t_1)} \rho(t_1) e^{ i \frac{H_S}{\hbar} (t-t_1)}$.   
\end{lemma} 
\begin{proof}
It follows immediately from Lemma \ref{lemmaequiv}.
\end{proof}
It is easy to see that a trivial interaction $V=0$ satisfies Condition \ref{cond1}.
The following condition is imposed to avoid such a trivial interaction.  
\begin{condition}\label{cond2}
(Non-triviality condition) 
An apparatus satisfies {\em the non-triviality condition} if   
there exists a state $\rho$ and a time $t>t_0$ such that 
\begin{eqnarray*}
[V, e^{-i \frac{H(t-t_0)}{\hbar} }(\rho \otimes \sigma(t_0))
e^{i \frac{H(t-t_0)}{\hbar}}]\neq 0. 
\end{eqnarray*}
\end{condition}
Note that this condition is rather weak.  
It still allows the existence of 
$\rho$ such that 
$V$ acts trivially 
on $\rho(t)\otimes \sigma(t)$ for any $t\in \mathbb{R}$. 
In addition, this condition does not specify exactly the switching time of the 
perturbation 
as $t=t_0$. It only requires the 
existence of the moment $(t>t_0)$ at which the perturbation does not vanish. 
\par
It is possible to further impose a condition on the switching-off process of an interaction. 
\begin{condition}\label{cond3}(Switching-off time condition) 
An apparatus satisfies {\em the switching-off time condition} if 
there exists a time $t_1 >t_0 $ called a switching off time such that 
for any $t\geq t_1$ and for any state $\rho $ of 
the system, 
\begin{eqnarray}
[V, e^{-i \frac{H (t-t_0)}{\hbar}} (\rho \otimes \sigma(t_0))
e^{i \frac{H (t-t_0)}{\hbar}}] =0  
\label{cond3eq}
\end{eqnarray}
holds. 
\end{condition}
One may consider a weaker definition for the switching-off time 
that depends on an initial state of the system. 
However, in this paper we do not use the switching-off condition. 
\begin{example}\label{ex1}
Suppose that an apparatus is described by 
$\hik=L^2(\mathbb{R})$ and 
$H_A =p:=-i \hbar \frac{d}{dx}$. 
It is coupled with a system described by 
$\hi=\mathbb{C}^2$ with orthonormal basis $\{|0\rangle, |1\rangle\}$ and 
a trivial Hamiltonian 
$H_S=0$. 
The interaction term is introduced as 
$V=\sigma_z\otimes g(q)$, 
where 
$\sigma_z=|1\rangle \langle 1|-|0\rangle \langle 0|$ and 
$g$ is a nonvanishing real function 
whose support is included in $(0, \Delta )$ for some $\Delta >0$. 
In the position representation $|\phi (0)\rangle$ is supposed to be strictly localized in 
the negative real line.  
That is, the support of $\phi(t=0, x):=\langle x|\phi(0)\rangle$ is included in 
$(-\delta,0)$ for some $\delta>0$. 
The freely evolved state of the apparatus can be written as 
$\phi(t,x)=\phi(0,x-t)$, which has support in $(-\delta+t, t)$. 
It is easy to see that this system-plus-apparatus satisfies Condition \ref{cond1}. 
\par
Let us consider the time evolution of the state 
denoted by 
$|\Omega(0)\rangle \otimes |\phi(0)\rangle$.
For $|\Omega(0)\rangle =|0\rangle $, 
while the state of the system remains $|0\rangle$, 
the state of the apparatus is affected by the interaction. 
We denote the 
whole state at time $t$ 
 by $|0\rangle \otimes |\phi^0(t)\rangle$. 
It can be shown that at time $t$, the state of the apparatus
$|\phi^0(t)\rangle$ becomes 
$\phi^0(t,x)=e^{-\frac{i}{\hbar}\int^x_0 dx' g(x')}\phi(t,x)$ in the position 
representation. 
On the other hand, 
for $|\Omega(0)\rangle =|1\rangle $, 
while the state of the system remains $|1\rangle$, 
the state $|\phi^1(t)\rangle$ of the apparatus at time $t$  
becomes 
$\phi^1(t,x)=e^{\frac{i}{\hbar}\int^x_0 dx' g(x')}\phi(t,x)$ in the position 
representation. 
Thus, for a general initial state 
$|\Omega(0)\rangle = c_0 |0\rangle + c_1|1\rangle$, 
the state of the total system evolves as
\begin{eqnarray*}
|\Phi(t)\rangle = c_0 |0\rangle \otimes |\phi^0(t)\rangle
+c_1 |1\rangle \otimes |\phi^1(t)\rangle, 
\end{eqnarray*}
which coincides with an unperturbed (freely evolved) state 
up to time $t=0$. Thus it satisfies Condition \ref{cond1}. 
The state of the system evolves as  
\begin{eqnarray*}
\rho(t)=
|c_0|^2 
|0\rangle \langle 0|+
|c_1|^2 |1\rangle \langle 1|
+c_0 \overline{c_1}|0\rangle \langle 1 | \langle \phi^1(t)|\phi^0(t)\rangle
+
c_1 \overline{c_0}
|1\rangle \langle 0 |\langle \phi^0(t)|\phi^1(t)\rangle, 
\end{eqnarray*}
which does not agree, in general, with the freely 
evolved state 
\begin{eqnarray*}
\rho^0(t)=
|c_0|^2 |0\rangle \langle 0|+
|c_1|^2
|1\rangle \langle 1|
+c_0 \overline{c_1}
|0\rangle \langle 1|+
c_1\overline{c_0}|1\rangle \langle 0|
\end{eqnarray*}
 for $t>0$.  
Thus Condition \ref{cond2} is also satisfied.
In addition, we put $\tau =\Delta + \delta$. 
Taking into account the support of $\phi(t,x)$, 
we obtain for $t>\tau$, 
\begin{eqnarray*}
\phi^0(t,x)=e^{-\frac{i}{\hbar} \int^{\Delta}_0 dx' g(x')}
\phi(t,x)\\
\phi^1(t,x)=e^{\frac{i}{\hbar}\int^{\Delta}_0 dx' g(x')}\phi(t,x).
\end{eqnarray*}
As their supports do not intersect with 
the support of $g$, Condition \ref{cond3} is satisfied. 
\end{example}
The total Hamiltonian in the above example has unbounded spectrum $\mathbb{R}$. This unboundedness  must be satisfied in general. 
\begin{theorem}
Let us consider a model satisfying both Condition \ref{cond1} and 
Condition \ref{cond2}. 
The spectrum of the total Hamiltonian $H= H_S +H_A +V$ 
is two-side unbounded (lower and upper unbounded). 
\end{theorem} 
\begin{proof}
Suppose that $H$ is lower bounded. (An upper bounded case can be treated similarly.) 
By purifying the state $\sigma(t_0)$ of the apparatus, 
we can assume the state $\sigma(t_0)$ to be pure. 
We denote the purified state by $\sigma(t_0)=|\phi(t_0)\rangle \langle \phi(t_0)|$. 
On this enlarged system, 
$H$ also works as a Hamiltonian and is lower bounded. 
(More precisely we need $H\otimes \id$, where $\id$ acts only on 
an auxiliary Hilbert space employed for purification.) 
We denote by $\hik$ again the enlarged Hilbert space of 
the apparatus.   
For an arbitrary state $|\Omega\rangle \in \hi$, 
we have for $t\leq t_0$, 
\begin{eqnarray*}
[V, e^{-i \frac{H}{\hbar}(t-t_0)}
|\Omega \otimes  \phi(t_0)\rangle 
\langle \Omega \otimes \phi(t_0) |
e^{i \frac{H}{\hbar}(t-t_0)}]=0. 
\end{eqnarray*}
This implies that for $t\leq t_0$ there exists a function $c(t)$ such that 
\begin{eqnarray*}
Ve^{-i \frac{H}{\hbar}(t-t_0)}|\Omega \otimes \phi(t_0)\rangle 
= c(t) |\Omega \otimes \phi(t_0)\rangle. 
\end{eqnarray*}
As $|\Omega \otimes \phi(t_0)\rangle $ is normalized we have 
\begin{eqnarray*}
Ve^{-i \frac{H}{\hbar}(t-t_0)}|\Omega \otimes \phi(t_0)\rangle 
= |\Omega \otimes \phi(t_0)\rangle 
\langle \Omega \otimes \phi(t_0)|V e^{-i \frac{H}{\hbar}(t-t_0)}
|\Omega \otimes \phi(t_0)\rangle. 
\end{eqnarray*}
For an arbitrary $|\Psi\rangle \in \hi\otimes \hik$ 
 we now define  
\begin{eqnarray*}
f_{\Psi,\Omega}(t):=\langle \Psi|
\left( V e^{-i \frac{H}{\hbar}(t-t_0)}
- \langle \Omega \otimes \phi(t_0) |Ve^{-i \frac{H}{\hbar}
(t-t_0)} |\Omega \otimes \phi(t_0)\rangle \right) 
|\Omega \otimes \phi(t_0)\rangle.  
\end{eqnarray*}
This is vanishing for $t\leq t_0$. 
As $H$ is lower bounded,  
$f_{\Psi,\Omega}(z):=\langle \Psi|
\left(Ve^{-i \frac{H_c}{\hbar}(z-t_0)} 
-  \langle \Omega \otimes \phi(t_0) |Ve^{-i \frac{H}{\hbar}
(z-t_0)} |\Omega \otimes \phi(t_0)\rangle \right) 
|\Omega \otimes \phi(0)\rangle $ can be defined for $Im(z)\leq 0$ and is analytic 
for $Im(z)<0$ \cite{Hergerfeldt}. 
The Schwarz reflection principle \cite{Ahlforsbook} 
implies that $f_{\Psi,\Omega}$ can be 
extended to an analytic function on 
$\mathbb{C}\setminus \{s|s > 0\}$. 
Because $f_{\Psi, \Omega}(t)=0$ on $t\leq t_0$, it follows 
that $f_{\Psi,\Omega}(t)=0$ for $t\in \mathbb{R}$.  
That is, for all $t\in \mathbb{R}$, 
\begin{eqnarray*}
Ve^{-i \frac{H}{\hbar}(t-t_0)}|\Omega \otimes \phi(t_0)\rangle 
= |\Omega \otimes \phi(t_0)\rangle 
\langle \Omega \otimes \phi(t_0)|V e^{-i \frac{H}{\hbar}(t-t_0)}
|\Omega \otimes \phi(t_0)\rangle 
\end{eqnarray*}
holds. This contradicts Condition \ref{cond2}. 
\end{proof}
The unbounded character of the total Hamiltonian implies that 
the total system is an unstable system or an infinite system. 
In the latter case, the Hamiltonian is the generator of time evolution 
referring to a state through the GNS representation \cite{Haagbook}. 
We present a more ``physical" example whose Hamiltonian 
is bounded from below.  In this example, the conditions are relaxed 
to yield small errors. 
\begin{example}\label{example:real}
Suppose that an apparatus is described by a one-particle 
Hilbert space $\hik = L^2(\mathbb{R})$. 
We consider a ``real" single particle which has 
Hamiltonian $H_A=\frac{p^2}{2m}=-\frac{\hbar^2}{2m}\frac{d^2}{dx^2}$ 
with a particle mass $m >0$.  
This $H_A$ is lower-bounded. 
We set $t_0 =0$ and some large $T\gg 0$.  
We set $|\phi(-T)\rangle$ in its position representation as
\begin{eqnarray*}
\phi(-T, x) = \frac{1}{\pi^{1/4}\sqrt{\sigma}}
\exp\left(
i k(x-x_0) - \frac{(x-x_0)^2}{2\sigma^2}
\right), 
\end{eqnarray*}
where $k >0$ and $x_0$ 
is chosen to satisfy 
$-x_0 = \frac{\hbar k}{m}(T+\Delta)$ for some fixed $\Delta >0$. 
Then it evolves according to the free Hamiltonian $H_A=\frac{p^2}{2m}$ as 
\begin{eqnarray*}
\phi(s-T, x)=
\frac{1}{\pi^{1/4} \sqrt{\sigma} 
\left( 1+\frac{i s\hbar}{\sigma^2 m}\right)}
\exp\left(
\frac{ - \frac{(x-x_0)^2}{2\sigma^2}+
ik(x-x_0) - \frac{i s \hbar k^2}{2m}}
{1 + \frac{is\hbar}{\sigma^2 m}}
\right). 
\end{eqnarray*} 
It gives rise to the probability distribution, 
for $v_g :=\hbar k/m$, 
\begin{eqnarray*}
|\phi(s-T,x)|^2 
= \frac{\pi^{1/2}}{\sigma 
\left(1+\frac{s^2 \hbar^2}{\sigma^4 m^2}
\right)}
\exp\left(
\frac{ - \frac{1}{\sigma^2}\left(x-x_0 -v_g s
\right)^2}
{1 +\frac{s^2\hbar^2}{\sigma^4 m^2}
} 
\right), 
\end{eqnarray*}
which gives an expectation and 
variance of $q$ as 
\begin{eqnarray}
\langle q(s-T)\rangle 
&=&x_0 +v_g s \label{av} \\
\Delta q(s-T)&:=&
\sqrt{\langle q(s-T)^2\rangle - \langle q(s-T)\rangle^2} = \sigma\left(
1 + \frac{\hbar^2 s^2}{\sigma^4 m^2}
\right)^{1/2}. \label{var}
\end{eqnarray}
Let us consider a composite system 
$\hi \otimes \hik$, where $\hi$ is a Hilbert space 
with a Hamiltonian $H_S$. 
The total Hamiltonian is assumed to have a form, 
\begin{eqnarray*}
H = H_S + H_A + B \otimes V(q),   
\end{eqnarray*}
where $B$ is a self-adjoint operator 
and $supp V(\cdot) \subset (0,\infty)$. 
One can see that for sufficiently large $t$ 
the interaction term gives a non-trivial effect. 
In the following we show that up to time $t_0=0$ the states 
evolve almost freely. 
We estimate, for $s\leq T$,  
\begin{eqnarray*}
\Vert 
e^{-i \frac{Hs}{\hbar} }|\xi\otimes \phi(-T)\rangle 
-  |\xi(t)\otimes \phi(s-T)\rangle\Vert
&=& 
\Vert |\xi\otimes \phi(-T)\rangle 
- W(s) |\xi\otimes \phi(-T)\rangle\Vert
\\
&\leq& 
\int^s_0 dt \left\Vert \frac{d}{dt}W(t)|\xi\otimes \phi(-T)\rangle 
\right\Vert, 
\end{eqnarray*}
where $W(t)= e^{i \frac{Ht}{\hbar}}
e^{-i \frac{H_0 t}{\hbar}}$ with $H_0 =H_S +H_A$. 
We obtain 
\begin{eqnarray*}
\frac{d}{dt}W(t) |\xi\otimes \phi(-T)\rangle 
= \frac{i}{\hbar}W(t) e^{i \frac{H_0 t}{\hbar}} (B\otimes V(q)) 
e^{-i \frac{H_0 t}{\hbar}}|\xi \otimes \phi(-T)\rangle. 
\end{eqnarray*}
Defining a projection $P_{\geq}$ by 
$(P_{\geq}\psi)(x)= \left\{ \begin{array}{cc} 
\psi(x)& (x\geq 0) \\
0 & (x < 0) \\
\end{array}
\right.$,  
we have 
\begin{eqnarray*}
e^{i \frac{H_0 t}{\hbar}}(B\otimes  V(q)) 
e^{-i \frac{H_0 t}{\hbar}}|\xi \otimes \phi(-T)\rangle 
= 
e^{i \frac{H_0 t}{\hbar}}(B\otimes V(q)) (\id \otimes P_{\geq})
e^{-i \frac{H_0 t}{\hbar}}|\xi \otimes \phi(-T)\rangle, 
\end{eqnarray*}
and 
\begin{eqnarray*}
&& \Vert 
e^{i \frac{H_0 t}{\hbar}}(B\otimes V(q)) 
e^{-i \frac{H_0 t}{\hbar}}|\xi \otimes \phi(-T)\rangle 
\Vert 
\leq \Vert B \otimes V(q) \Vert 
\Vert (\id \otimes P_{\geq}) e^{-i \frac{H_0 t}{\hbar}}|\xi \otimes \phi(-T)\rangle 
\Vert
\\
&=&
\Vert B \Vert \Vert V(q)\Vert 
\Vert P_{\geq}e^{- \frac{H_A t}{\hbar}} |\phi(-T)\rangle\Vert.
\end{eqnarray*}
We can estimate  
$\Vert P_{\geq} |\phi(t-T) \rangle \Vert^2 
= \int^{\infty}_0 |\phi(t-T, x)|^2$. 
We use (\ref{av}), (\ref{var}) and the
Chebyshev inequality to give, for $t\leq T$,  
\begin{eqnarray*}
\Vert P_{\geq} |\phi(t-T) \rangle \Vert^2 
\leq \frac{
\sigma^2 \left( 1 + \frac{\hbar^2 t^2}{\sigma^4 m^2}\right)}
{(x_0 +v_g t)^2}
\leq \frac{\sigma^2 \left(1+\frac{\hbar^2 t^2}{\sigma^4 m^2}\right)}
{(v_g \Delta)^2}, 
\end{eqnarray*}
which can be made arbitrarily small as 
$k \to \infty$. 
Thus, we have
\begin{eqnarray*}
\Vert 
e^{-i \frac{Hs}{\hbar} }|\xi\otimes \phi(-T)\rangle 
-  |\xi(t)\otimes \phi(s-T)\rangle\Vert
\leq
\Vert B\Vert \Vert V(q)\Vert 
\frac{\sigma}{v_g \Delta} 
\left( s + \frac{\hbar^2}{3 \sigma^4 m^2} s^3
\right),
\end{eqnarray*}
which can be arbitrary small for $k\to \infty$. 
Thus, up to time $t_0=0$, the interaction can be made 
negligibly small. 
%
\end{example} 
\section{Energy-Time Uncertainty Relation I: Energy of Apparatus}
\label{section4}
In this section we investigate how much energy is required  
to realize a measurement process satisfying Condition \ref{cond1}. More precisely, 
we study the energy fluctuation $\Delta H_A$ of the apparatus Hamiltonian 
required for a measurement within  
a measurement time duration $\tau$, and derive the  
energy-time uncertainty relation between them. 
\par
We treat a closed composite system consisting of
 a system and a measuring apparatus whose total Hilbert space is 
$\hi \otimes \hik$. 
The dynamics is governed by a total Hamiltonian 
$H= H_S + H_A + V$, 
where $H_S$ (respectively $H_A$) is the system Hamiltonian 
acting on $hi$ (resp. $\hik$) and 
$V$ is the interaction.
The apparatus must be large enough to include 
the switching-on process. That is, the model with an ``initial" state 
$\sigma(t_0)$ of 
the apparatus at time $t=t_0$ satisfies Condition \ref{cond1}. 
In addition, the model is assumed to describe a measurement process 
of a sharp observable i.e., a PVM $\P=\{\P_n\}_{n \in \Omega_{\P}}$ with a measurement 
time duration $\tau >0$. 
That is, there exists a POVM $\E=\{\E_n\}_{n \in \Omega_{\P}}$ 
acting on $\hik$ such that 
it holds that for an arbitrary state $\rho$ of the system, 
\begin{eqnarray*}
\mbox{tr}[\rho \P_n] 
= \mbox{tr}[e^{-i \frac{H}{\hbar}\tau} 
(\rho \otimes \sigma(t_0))e^{i \frac{H}{\hbar}{\tau}}(\id \otimes \E_n)]. 
\end{eqnarray*}  
\par
Two observations play roles in deriving a result. 
The first is a so-called information-disturbance trade-off 
relation in the quantum measurement process. 
In quantum mechanics, 
any non-trivial measurement process causes 
disturbance of quantum states. 
This property 
is directly related to the uncertainty relation of joint measurement. 
In fact, the disturbance caused by a measurement of an observable 
spoils the subsequent measurements. 
In particular, a perfectly accurate measurement of an observable 
implies a complete distortion of the conjugate states. 
The second observation, which is derived in this section for the 
first time, is a kind of energy-time uncertainty relation 
between the energy fluctuation of an apparatus and the time duration 
required to disturb a state of the system. 
This relation is proved by applying the Mandelstam-Tamm 
uncertainty relation to a timing device satisfying Condition \ref{cond1}.  
Combining these two observations, we derive a trade-off 
inequality between the time duration of the measurement process 
and the energy fluctuation of the apparatus. 
\par
We begin with the latter problem: how large an energy 
fluctuation is required to disturb a state. 
We examine the behavior of states of a system 
interacting with a timing device. 
We consider the 
time evolution of 
$
 \Theta_0(t_0):=
\rho(t_0)\otimes \sigma(t_0)$, 
where $\rho(t_0)$ is a state of the system at time $t=t_0$. 
Its evolved state is denoted by $\Theta_0(t)=e^{-i\frac{H}{\hbar}(t-t_0)
}\Theta_0(t_0)
e^{i\frac{H}{\hbar}(t-t_0)}$, and its restricted state on the system 
is written as $\rho(t):=\mbox{tr}_{\hik}\Theta_0(t)$, 
where $\mbox{tr}_{\hik}$ represents 
a partial trace.  
This state $\rho(t)$ should be compared to  
the freely evolved (imaginary) state $\rho^0(t):=e^{-i\frac{H_S}{\hbar}(t-t_0)}
\rho(t_0)e^{i\frac{H_S}{\hbar}(t-t_0)}$. 
\par 
We employ a quantity, called fidelity, to quantify the deviation.  
The fidelity between the two states $\rho_0$ and $\rho_1$ is defined by 
$F(\rho_0,\rho_1):=\mbox{tr}(\sqrt{\rho_0^{1/2}\rho_1\rho_0^{1/2}})$. 
It satisfies $0\leq F(\rho_0,\rho_1)\leq 1$, and 
$F(\rho_0,\rho_1)=1$ holds if and only if $\rho_0=\rho_1$. 
 (See for example, \cite{Nielsen}.) 
\par 
Now we consider a timing device which satisfies Condition \ref{cond1}.  
The energy fluctuation of the apparatus at $t=t_0$  
\begin{eqnarray*}
\Delta H_A :=\sqrt{\mbox{tr}[\sigma(t_0) H_A^2] -
\mbox{tr}[\sigma(t_0) H_A]^2}
\end{eqnarray*}
plays a central role in the following theorems.
Because up to time $t=t_0$ the apparatus evolves only by 
$H_A$, the energy fluctuation is constant for $t\leq t_0$.  
This quantity characterizes the speed of the time evolution for the isolated apparatus. 
Because of Condition \ref{cond1}, the state of the apparatus evolves up to time $t=t_0$ 
as if it is isolated.  
Thus $\sigma(t)$, a state of the apparatus for $t\leq t_0$, is written as 
$\sigma(t) = e^{-i \frac{H_A}{\hbar}(t-t_0)}\sigma(t_0) 
e^{i \frac{H_A}{\hbar}(t-t_0)}$.  
The well-known Mandelstam-Tamm time energy uncertainty 
relation \cite{Mandelstam-Tamm, Busch1, BuschTEUR, Pfeifer} 
states that for any normalized vector 
$|\phi(t)\rangle =e^{-i 
\frac{H_A}{\hbar}(t-t_0)}|\phi(t_0)\rangle \in \hik$, 
their overlap is bounded as 
\begin{eqnarray*}
|\langle \phi(t)| \phi(t_0)\rangle | 
\geq  \cos \left( \frac{\Delta H_A (t-t_0)}{\hbar}
\right) 
\end{eqnarray*}
for $|t-t_0| \leq \frac{\pi \hbar}{2 \Delta H_A}$.  
\begin{lemma}
For $t$ satisfying $0\leq t_0 -t \leq \frac{\pi \hbar}{2\Delta H_s}$, 
the fidelity between $\sigma(t)$ and $\sigma(t_0)$ is bounded as 
\begin{eqnarray}
F(\sigma(t), \sigma(t_0))
\geq \cos\left(
\frac{\Delta H_A (t-t_0)}{\hbar}
\right). 
\label{Mandelstam}
\end{eqnarray}
\end{lemma}
\begin{proof}
The fidelity is the maximum overlap 
between the purified states as $F(\sigma(t), \sigma(t_0))
= \sup |\langle \phi(t) |\phi(t_0)\rangle|$, and the 
above inequality can be obtained for the fidelity by using $\Delta H_A 
= \Delta (H_A \otimes \id)$, where $\Delta(H_A \otimes \id)$ is 
the energy fluctuation of the purified state $|\phi(t_0)\rangle$ 
of $\sigma(t_0)$.
\end{proof}  
The deviation of the perturbed state, measured by the fidelity, can be 
bounded by the following simple argument. 
\begin{theorem}\label{ETtheorem}
Let us consider a process that satisfies Condition \ref{cond1}.
$F(\rho(t), \rho^0(t))$ represents the fidelity between $\rho(t)$ and 
the freely evolved state $\rho^0(t)$. 
It holds that for $0\leq t-t_0\leq \frac{\pi \hbar}{2\Delta H_A }$,   
\begin{eqnarray*}
\cos \left( \frac{\Delta H_A (t-t_0)}{\hbar}\right) \leq F(\rho(t), \rho^0(t)). 
\end{eqnarray*}
\end{theorem}
\begin{proof}
Let us consider two states $\Theta_{t_0}(t_0)$ and $\Gamma_t(t_0)$ $(0\leq t-t_0
\leq \frac{\pi \hbar}{2 \Delta H_c})$ defined by 
\begin{eqnarray*}
\Theta_{t_0}(t_0)&:=&\rho(t_0)\otimes \sigma(t_0)
\\
\Gamma_t(t_0)&:=&\rho(t_0) \otimes \sigma(t_0 -(t-t_0)).
\end{eqnarray*}
These states evolve with the Hamiltonian $H=H_S+H_A+V$. 
We denote the states at time $t$ by $\Theta_{t_0}(t)$ and $\Gamma_t(t)$. 
Although $\Theta_{t_0}(t)$ may have a complicated form, $\Gamma_t(t)$ has 
a simple form as it evolves freely, 
\begin{eqnarray*}
\Gamma_t(t)=\rho^0(t)\otimes \sigma(t_0),
\end{eqnarray*}
where we used 
Lemma \ref{lemma1}.
Because the fidelity between the two states is invariant 
under unitary evolution \cite{Nielsen}, it holds that 
\begin{eqnarray*}
F(\Theta_{t_0}(t_0), \Gamma_t(t_0))
=F(\Theta_{t_0}(t), \Gamma_t(t)). 
\end{eqnarray*}
The left-hand side of the above equation becomes 
\begin{eqnarray*}
F(\Theta_{t_0}(t_0), \Gamma_t(t_0))
=
F(\sigma(t_0), \sigma(t_0 -(t-t_0)), 
\end{eqnarray*}
and the right-hand side is bounded as 
\begin{eqnarray*}
F(\Theta_{t_0}(t), \Gamma_t(t))\leq F(\rho(t), \rho^0(t)), 
\end{eqnarray*}
where we utilized the fact that the fidelity decreases for restricted states \cite{Nielsen}. 
Thus it holds that 
\begin{eqnarray*}
F(\sigma(t_0), \sigma(t_0 -(t-t_0)) \leq F(\rho(t),\rho^0(t)). 
\end{eqnarray*} 
The left-hand side of this inequality represents the speed of 
time evolution of the apparatus and is bounded as (\ref{Mandelstam}).
This concludes the proof. 
\end{proof}
\par
Now, we use an information-disturbance trade-off relation 
to estimate the size of the 
perturbation caused 
by a measurement. 
We 
study a process in which a sharp observable is perfectly measured.  
The following is our main result: 
\begin{theorem}\label{th:mainmain}
Let us consider a measurement process of a sharp observable $\P
=\{\P_n\}$ that satisfies  
Condition \ref{cond1}.
Its measurement time duration $\tau$ and energy fluctuation of 
an apparatus $\Delta H_A$ must satisfy 
\begin{eqnarray*}
\tau \cdot \Delta H_A \geq \frac{\pi \hbar}{4}. 
\end{eqnarray*}
\end{theorem} 
\begin{proof}
We consider the dynamics  
from time $t=t_0$ to $t=t_0+\tau$. 
Suppose that $0$ and $1$ are possible outcomes. 
We consider two states $|0\rangle $ and $|1\rangle$ 
satisfying $P_0 |0\rangle = |0\rangle$, $P_1|0\rangle  = 0$, 
$P_1|1\rangle = |1\rangle$, and $P_0|1\rangle =0$.
We suppose that for a time duration $\tau$ 
an initial state (at time $t_0$) $|0\rangle \langle 0|\otimes \sigma(t_0)$ 
evolves to $\Theta_0$, 
and an initial state $|1\rangle \langle 1|\otimes\sigma(t_0)$ 
evolves to $\Theta_1$. 
In a perfectly accurate measurement process, 
their restricted states to the apparatus, 
$\mbox{tr}_{\hi}[\Theta_0]$ and $\mbox{tr}_{\hi}[\Theta_1]$,  
are perfectly distinguishable.
We can show that this dynamics spoils completely 
the possibility to measure 
its ``conjugate" observable afterwards. 
That is, if we define $|\pm\rangle = \frac{1}{\sqrt{2}}
(|0\rangle \pm |1\rangle)$, 
initial states $|\pm\rangle \langle \pm | \otimes \sigma(t_0)$ 
evolve to states whose restriction to the 
system coincide. 
We follow the derivation given by Janssens and Maassen \cite{Janssens}. 
The states to be compared are 
$\mbox{tr}_A[e^{-i \frac{H}{\hbar}\tau} 
(|\pm\rangle \langle \pm| \otimes \sigma(t_0))
e^{i \frac{H}{\hbar}\tau}]$. 
For an arbitrary operator $\A$ on the system, 
it holds that 
\begin{eqnarray}
|\mbox{tr}[e^{-i \frac{H}{\hbar}\tau} ((|+\rangle \langle + |
- |-\rangle \langle - |) \otimes \sigma(t_0)) e^{i \frac{H}{\hbar}\tau}
(\A\otimes \id)]
&\leq& 2 | \mbox{tr}[e^{-i \frac{H}{\hbar}\tau}
(|1\rangle \langle 0 | \otimes \sigma(t_0))e^{i \frac{H}{\hbar}\tau} 
(\A \otimes \id)]| 
\nonumber \\
&\leq &
2 \sum_n |\mbox{tr}[e^{-i \frac{H}{\hbar}\tau}
(|1\rangle \langle 0 | \otimes \sigma(t_0))e^{i \frac{H}{\hbar}\tau} 
(\A \otimes \E_n)]| 
\nonumber \\
&=&
2 \sum_n | \langle 0 \otimes \phi(t_0)| e^{i \frac{H}{\hbar}\tau}
(\A\otimes \E_n\otimes \id) e^{-i \frac{H}{\hbar}\tau} |1\otimes \phi(t_0)\rangle |,  
\label{tochuu1}
\end{eqnarray}
where $|\phi(t_0)\rangle $ is a purified vector of 
$\sigma(t_0)$. 
By using the Cauchy-Schwarz inequality, we obtain 
 \begin{eqnarray*}
(\ref{tochuu1})
&\leq& 2 \Vert \A\Vert 
\sum_n \langle 0 \otimes \phi(t_0)| 
e^{i \frac{H}{\hbar}\tau} (\id \otimes \E_n\otimes \id)
e^{-i \frac{H}{\hbar}\tau}|0 \otimes \phi(t_0)\rangle^{1/2}
\langle 1 \otimes \phi(t_0)| 
e^{i \frac{H}{\hbar}\tau} (\id \otimes \E_n\otimes \id)
e^{-i \frac{H}{\hbar}\tau}|1 \otimes \phi(t_0)\rangle^{1/2}
\\
&=&
2 \Vert \A \Vert \sum_n \delta_{n0} \delta_{n1} =0. 
\end{eqnarray*}
Thus the states $|+\rangle \langle +| \otimes \sigma(t_0)$ 
and $|-\rangle \langle -| \otimes \sigma(t_0)$ 
evolve into states whose restriction to the system is 
an identical state $\rho$. 
One can conclude that at least one of the states 
$|\pm\rangle$ is perturbed.  
To estimate the magnitude of this perturbation we consider 
unitary evolution governed by the Hamiltonian $H_S$. 
In time $\tau$, this ``unperturbed" dynamics changes 
$|\pm\rangle $ to a pair of orthogonal 
states $|\pm'\rangle$ of the system. 
We estimate $F(\rho, |+'\rangle \langle +'|)$ 
and $F(\rho, |-'\rangle \langle -'|)$. 
As $|\pm'\rangle$ are orthogonal, we have  
\begin{eqnarray*}
F(\rho, |+'\rangle \langle +'|)^2 + F(\rho, |-'\rangle \langle -'|)^2 
&=&\langle +' |\rho|+'\rangle + \langle -'| \rho |-'\rangle 
\\
&\leq& \mbox{tr}[\rho]=1. 
\end{eqnarray*} 
Thus we can conclude 
\begin{eqnarray*}
\min \{ F(\rho, |+'\rangle \langle +'|), F(\rho, |-'\rangle \langle -'|)\} 
\leq \frac{1}{\sqrt{2}}. 
\end{eqnarray*}
We assume $F(\rho, |+'\rangle \langle +'|)\leq \frac{1}{\sqrt{2}}$. 
Combining it with Theorem \ref{ETtheorem} 
by putting $\rho(t_0)=|+\rangle \langle +|$, 
we obtain 
\begin{eqnarray*}
\cos\left(
\frac{\tau \Delta H_A}{\hbar}  
\right)\leq \frac{1}{\sqrt{2}}.
\end{eqnarray*}
This ends the proof. 
\end{proof}
The above theorem can be easily extended to obtain 
the following:
\begin{proposition}
Let us consider a measurement process satisfying Condition \ref{cond1}
of 
a sharp observable $\P$ that has at least $N$ outcomes. 
The measurement time duration $\tau$ and the fluctuation of the apparatus 
Hamiltonian $\Delta H_A$ satisfy 
\begin{eqnarray*}
\cos\left(\frac{\tau \Delta H_A}{\hbar}\right)
\leq \frac{1}{\sqrt{N}}.
\end{eqnarray*}
\end{proposition}
\begin{proof}
We can assume that PVM $\P$ 
has the elements $\{\P_0, \P_1, \ldots, \P_{N-1}\}$
that have eigenvectors $\{|n\rangle\}_{n=0,1,\ldots, N-1}$ satisfying 
$\P_n |n\rangle =|n\rangle$. 
Then one can show that states $|\tilde{k}\rangle 
:= \frac{1}{\sqrt{N}}\sum_n e^{i \frac{ 2\pi k n}{N}}|n\rangle$ 
evolve into states whose restrictions to the system coincide.
In fact, for an arbitrary operator $\A$ on the system, it holds that
for $m \neq l$,  
\begin{eqnarray*}
| \mbox{tr}[e^{-i \frac{H}{\hbar}\tau}(|\tilde{m}\rangle 
\langle \tilde{m} | - |\tilde{l} \rangle 
\langle \tilde{l}| \otimes \sigma(t_0) )e^{i 
\frac{H}{\hbar} \tau}(\A\otimes \id)] 
\leq \frac{1}{N} 
\sum_{n \neq n'} | \langle n \otimes \phi(t_0) | 
e^{i \frac{H}{\hbar}\tau}(\A \otimes \id ) 
e^{-i \frac{H}{\hbar}\tau}|n' \otimes \phi(t_0)\rangle |, 
\end{eqnarray*}
where $|\phi(t_0)\rangle$ is a purified state of $\sigma(t_0)$. 
The argument employed above bounds the right hand side of this 
inequality by zero. 
\par
Let us denote this identical state at time $t=t_0+ \tau$ on the system by $\rho$. 
We denote the states obtained by unperturbed (free) 
time evolution by $\{|\tilde{k}'\rangle\}$, which are orthonormal.  
Now we have
\begin{eqnarray*}
\sum_{k=0}^{N-1} F(\rho, |\tilde{k}'\rangle \langle \tilde{k}'|)^2
\leq \mbox{tr}[\rho]=1. 
\end{eqnarray*} 
Thus $min_k F(\rho, |\tilde{k}'\rangle \langle \tilde{k}'|) \leq \frac{1}{\sqrt{N}}$ 
is obtained. This ends the proof. 
\end{proof}
Letting $N$ go to $\infty$, we obtain the following: 
\begin{corollary}
Let us consider a measurement process satisfying Condition \ref{cond1} of 
a sharp observable that has infinitely many outcomes.
The measurement time duration $\tau$ and the fluctuation of the apparatus 
Hamiltonian $\Delta H_A$ must satisfy 
\begin{eqnarray*}
\tau \cdot \Delta H_A \geq \frac{\pi \hbar}{2}. 
\end{eqnarray*}
\end{corollary}
As $\Delta H_A$ can be infinite for some states, it is instructive to 
derive a similar trade-off relation for other quantities characterizing 
energy indefiniteness. 
For $0 < \alpha \leq 1$, the overall width $\Delta_{\alpha}f$ of a distribution function $f$ on $\mathbb{R}$ is 
defined as the width of the smallest interval $I$ such that 
\begin{eqnarray*}
\int_I  dx  f(x) \geq \alpha. 
\end{eqnarray*}
This quantity has its own corresponding energy-time 
uncertainty relation 
\cite{Uffink}. 
A closed system with a Hamiltonian $H$ is considered. 
Let $\tau_\beta$ be the minimal time it takes for a
state $|\psi(0)\rangle$ to evolve to a state 
$|\psi(t)\rangle$ such that 
\begin{eqnarray*}
|\langle \psi(0)| \psi(t)\rangle | \leq \beta. 
\end{eqnarray*} 
Then one can show, for $\beta\leq 2\alpha-1$,
\begin{eqnarray*}
\tau_{\beta} \cdot \Delta_{\alpha}(H) \geq 
2\hbar \arccos \frac{\beta + 1-\alpha}{\alpha}, 
\end{eqnarray*}
where $\Delta_{\alpha}(H)$ is the overall width 
of a distribution function $f(E)$ defined by 
$\langle \psi(0)| E(dE) |\psi(0)\rangle 
= f(E) dE$, where $E(dE)$ is a spectral measure of $H$. 
\par 
This inequality is employed to derive another version of 
the above result. The proof is given in the Appendix 
\ref{AppendixA}. 
\begin{theorem}\label{theoremappendix}
Let $\tau$ be a measurement duration. For 
any $\alpha$ with $\alpha\geq \frac{1}{2}\left(1 +\frac{1}{\sqrt{2}}\right)$, 
it holds that 
\begin{eqnarray*}
\tau \cdot \Delta_{\alpha}(H) 
\geq 2 \hbar \arccos \frac{\frac{1}{\sqrt{2}}+ 1 -\alpha}{\alpha}.
\end{eqnarray*}
\end{theorem}
It is interesting to extend the results to measurements with errors.
Let us consider again a sharp observable $\P=\{\P_n\}$. 
For an initial state 
$\rho_n$ with $\P_n \rho_n \P_n =\P_n$,
an imperfect measurement allows errors. 
That is, there may exist $\rho_n$ and $m$ such that   
$Prob(m|\rho_n):=
\mbox{tr}[e^{-i \frac{H \tau}{\hbar}} (\rho_n \otimes \sigma(t_0))
e^{i \frac{H \tau}{\hbar}} (\id \otimes \E_n)]
\neq \delta_{nm}$ holds.  
We introduce 
$P_{error}$ by 
\begin{eqnarray*}
P_{error}:= \sup_n \sup_{\rho_n: \P_n \rho_n \P_n = 
\rho_n} (1- P(n|\rho_n)) 
\end{eqnarray*}
which represents 
the worst case error probability.
\begin{theorem}\label{th:error}
Let us consider a measurement process satisfying 
Condition \ref{cond1} of a sharp observable $\P
=\{\P_n\}$ with the worst case error probability $P_{error}$.
Its measurement time duration $\tau$ and energy fluctuation of 
an apparatus $\Delta H_A$ satisfies 
\begin{eqnarray*}
\cos\left(
\frac{\tau \Delta H_A}{\hbar}
\right)
\leq 
\sqrt{
\frac{1+ 6 \sqrt{P_{error}}}{2}
}.
\end{eqnarray*}
\end{theorem} 
The proof is given in Appendix \ref{sec:error}. 

As mentioned, the uncertainty relation for joint measurement, or 
the information-disturbance relation, plays a crucial role in our proof. 
While the quantitative relation given by Janssens and Maassen was employed 
here, other quantitative relations 
\cite{MiHeisenberg, MiJMP, HeMi} 
ought to be applicable. It would be interesting 
to compare the corresponding energy-time uncertainty relations. 
\begin{example} 
Let us consider an apparatus consisting of 
a particle in a two-dimensional $xz-$plane. Its Hilbert space 
is written as $\hik = L^2(\mathbb{R}^2)$. 
We denote by $(q_x, p_x, q_z, p_z)$ 
position and momentum operators 
of their corresponding subscripts. 
We set an apparatus Hamiltonian $H_A$ as 
\begin{eqnarray*}
H_A = p_x. 
\end{eqnarray*}
We consider a qubit (spin-$1/2$ degree of freedom) as a system
whose self Hamiltonian $H_S$ is vanishing. 
Thus a total Hilbert space is 
written as $\hi_T = \hi \otimes \hik
= \mathbb{C}^2 \otimes L^2(\mathbb{R}^2)$. 
We study a time evolution determined by 
the total Hamiltonian,
\begin{eqnarray*}
H= V + H_A =  \sigma_z g(q_x) p_z + p_x, 
\end{eqnarray*}
where $g$ is a nonnegative smooth function whose support is 
in $(0,\delta)$ with $\delta >0$. 
We set an initial state (at time $t=0$) of the apparatus in the position representation 
as
\begin{eqnarray*}
\phi(t=0,x,z)=\langle x, z| \phi(0)\rangle 
= \xi(x) \eta(z), 
\end{eqnarray*}
where $\xi$ is a smooth real function satisfying $\mbox{supp} \xi \subset (-\Delta, 0)$ 
for some $\Delta >0$ and 
$\eta$ is a smooth function 
with $\eta(-z)=\eta(z)$ and $\mbox{supp} \eta \subset (-\epsilon, \epsilon)$ for some 
$\epsilon >0$. 
We set $\tau := \delta + \Delta$ and 
assume that $\epsilon >0$ is small enough to satisfy 
\begin{eqnarray*}
\epsilon < \int^{\tau}_0 dt g(t). 
\end{eqnarray*} 
It is easy to see that this model satisfies Condition \ref{cond1}.  
Its energy fluctuation up to time $t=0$ is 
\begin{eqnarray*}
\Delta H_A = - \hbar^2 \int dx \overline{\xi (x)} \frac{d^2}{dy^2} \xi (x) 
= \hbar^2 \int dx \left(
\frac{d \xi(x)}{dx} 
\right)^2,  
\end{eqnarray*}
which becomes finite for sufficiently smooth $\xi$.
After a lengthy calculation one can see that 
if an initial state of the system is $|1\rangle$ (respectively $|-1\rangle$), 
the probability for $q_z >0$ (resp. $q_z <0$) at time $t=\tau$ is $1$. 
Thus this model works as a measurement model of $\sigma_z$. 
If we rescale $\xi$ as, for $C>0$,  
\begin{eqnarray*}
\xi_C (x):= \sqrt{C} \xi(Cx), 
\end{eqnarray*}
$\Delta$ is scaled as $\Delta /C$ and thus 
$\tau$ scales as $\tau /C$.  
On the other hand, the energy fluctuation $\Delta H_A$ is  
scaled as $C \Delta H_A$. Thus this model illustrates 
the expected energy-time trade-off.  
\end{example}
We close this section with a remark on the definition of apparatus. 
The above example can be regarded as a toy model of the Stern-Gerlach experiment 
if the qubit is interpreted as the spin degree of freedom of a moving particle. 
Thus, the apparatus in this model is the position degree of freedom of a particle. 
\section{Energy-Time Uncertainty Relation II: Interaction Strength}\label{section5}
In the previous section, we derived a trade-off inequality 
between the fluctuation of $H_A$ and the measurement time duration $\tau$. 
In this section, we examine the magnitude of interaction $V$ 
that is required for measurements.  
In this section Condition 
\ref{cond1} 
is not assumed. 
The measurement process is described by a total Hamiltonian 
$H= H_S + H_A +V$ and 
the dynamics up to time $t=t_0$ is irrelevant.  
A state of the apparatus $\sigma(t_0)$ is prepared at time 
$t=t_0$ and 
 a meter observable $\E$ is measured 
at $t=t_0+\tau$.  
The reason why 
we do not need conditions for switching-on device to 
obtain the bound on the interaction is that 
the interaction Hamiltonian is mainly relevant to 
the border between the system and the apparatus while the 
switching-on condition is relevant to the border between the apparatus 
and the classical system.  
\begin{theorem}
For a model to describe a measurement process of 
a PVM, 
the interaction $V$ must satisfy  
\begin{eqnarray*}
\Vert V \Vert \cdot \tau 
\geq \frac{\pi}{4}\hbar. 
\end{eqnarray*}
\end{theorem}
\begin{proof}
Again we consider a process exactly measuring   
a PVM $\P$. 
Two states $|0\rangle $ and $|1\rangle$ satisfying 
$P_0 |0\rangle = |0\rangle$, $P_1|1\rangle = |1\rangle$ 
and $P_1|0\rangle =0$, $P_0 |1\rangle =0$  
exist.  
For a process to distinguish these states 
exactly, its conjugate states are completely 
destroyed. 
That is, if we define $|\pm\rangle = \frac{1}{\sqrt{2}}(|0\rangle + |1\rangle )$, 
initial states $|\pm\rangle \langle \pm | \otimes \sigma(t_0)$ 
are mapped to states that are an identical on the system side. 
We denote by $\rho$ such a final state on the system. 
We also denote by $|\pm'\rangle $ the states obtained by 
free evolution without interaction in time $\tau$.  
We have already shown that
\begin{eqnarray*}
\min\{F(\rho, |+'\rangle \langle +'|), F(\rho, |-' \rangle \langle -'|)\} 
\leq \frac{1}{\sqrt{2}}.
\end{eqnarray*}
We assume $F(\rho, |+ ' \rangle \langle +'|)
\leq \frac{1}{\sqrt{2}}$. 
The state $\rho$ is written as 
\begin{eqnarray*}
\rho = \mbox{tr}_{A} 
[
e^{-i \frac{H \tau}{\hbar}}
(|+\rangle \langle +|\otimes |\phi(t_0)\rangle \langle \phi(t_0)|)
e^{i \frac{H \tau}{\hbar}}
], 
\end{eqnarray*}
where $|\phi(t_0)\rangle$ is a purified state of $\sigma(t_0)$ 
and $\mbox{tr}_A$ denotes a partial trace to obtain 
a restricted state on the system. 
We introduce a time-evolving vector 
$|\Phi (t)\rangle:= e^{-i\frac{H (t-t_0)}{\hbar}} 
|+\rangle \otimes |\phi(t_0)\rangle$. 
On the other hand, 
$|+'\rangle$ satisfies 
\begin{eqnarray*}
|+'\rangle \langle +' |\otimes |\phi(t_0+\tau)\rangle 
\langle \phi(t_0+\tau)| 
= e^{- i \frac{H_0 \tau}{\hbar}}|+\rangle \langle +| \otimes 
|\phi(t_0)\rangle \langle \phi(t_0) | e^{i \frac{H_0 \tau}{\hbar}}, 
\end{eqnarray*}
where $H_0 = H_S+ H_A$. 
We introduce $|\Phi_0(t)\rangle := e^{-i \frac{H_0 (t-t_0)}{\hbar}}
|+\rangle \otimes |\phi(t_0)\rangle$. 
Since the fidelity does not decrease by 
the partial trace, 
the vectors $|\Phi(t)\rangle$ and $|\Phi_0(t)\rangle$ satisfy
\begin{eqnarray}
|\langle \Phi_0(t_0+\tau) |\Phi(t_0 +\tau)\rangle | \leq \frac{1}{\sqrt{2}}. 
\label{eqn12}
\end{eqnarray}
Defining $p(t):= \langle \Phi_0(t) | \Phi(t) \rangle \langle \Phi(t) |\Phi_0(t)\rangle$, 
we obtain 
\begin{eqnarray*}
i \hbar \frac{d}{dt} p(t) 
= \langle \Phi_0(t) | [V, |\Phi(t) \rangle \langle \Phi(t) | ] \Phi_0(t)\rangle. 
\end{eqnarray*}
The right-hand side of this equation can be further bounded 
by the Robertson uncertainty relation as
\begin{eqnarray*}
| \langle \Phi_0(t) |[V, |\Phi(t)\rangle \langle \Phi(t)|] |\Phi_0(t)\rangle |
\leq 2 (\Delta V)_t \sqrt{p(t) -p(t)^2},   
\end{eqnarray*}
where $(\Delta V)_t := 
\sqrt{\langle \Phi_0(t) |V^2 |\Phi_0(t)\rangle - \langle \Phi_0(t) |V|\Phi_0(t)\rangle^2}
\leq \Vert V \Vert$. 
Thus we obtain, 
\begin{eqnarray*}
\frac{d}{dt} p(t) \leq 2 \Vert V \Vert \sqrt{p(t) -p(t)^2}, 
\end{eqnarray*}
whose solution shows that
\begin{eqnarray}
p(t) \geq \cos^2 \left( \frac{\Vert V\Vert}{\hbar} (t-t_0)\right). 
\label{eqn13} 
\end{eqnarray}
By combining (\ref{eqn12}) and (\ref{eqn13}), 
we conclude the proof. 
\end{proof}
As in the case of energy fluctuation, the 
right-hand side of the trade-off relation 
gets larger depending on the number of outcomes. 
In particular, for observables which have infinitely many outcomes 
we have the following theorem: 
\begin{theorem}
Let us consider a measurement model 
that measures a sharp 
observable with infinitely many outcomes. 
Its interaction $V$ must satisfy 
\begin{eqnarray*}
\Vert V\Vert \cdot \tau \geq \frac{\pi}{2} \hbar. 
\end{eqnarray*} 
\end{theorem}
\begin{proof}
For an arbitrary $N>0$ we can assume that the PVM to be measured has 
the elements $\P_1, \P_2,\ldots, \P_N$
with eigenvectors $\{|n\rangle\}_{n=1,2, \ldots, N}$ satisfying 
$\P_n |n\rangle = |n\rangle$. 
(If $\P$ has a continuous outcome set, 
we construct the above elements by discretizing the set.)  
We can show that the states $|\tilde{k}\rangle 
:= \frac{1}{\sqrt{N}}\sum_n e^{i \frac{ 2\pi k n}{N}}|n\rangle$ 
are mapped to an identical state, say $\rho$, on the system.  
We denote the states obtained by unperturbed (free) 
time evolution by $\{|\tilde{k}'\rangle\}$, which are orthonormalized.  
Now we have,  
\begin{eqnarray*}
\sum_{k=1}^{N } F(\rho, |\tilde{k}'\rangle \langle \tilde{k}'|)^2
\leq \mbox{tr}[\rho]=1. 
\end{eqnarray*} 
Thus, $min_k F(\rho, |\tilde{k}'\rangle \langle \tilde{k}'|) \leq \frac{1}{\sqrt{N}}$ 
is obtained. 
An argument in the previous proof is applied to show
\begin{eqnarray*} 
\cos \left( \frac{\Vert V\Vert}{\hbar} \tau \right) 
\leq \frac{1}{\sqrt{N}}. 
\end{eqnarray*}
As $N$ was arbitrary, it ends the proof.  
\end{proof}
In \cite{MiyaderaInteraction2011}, a similar bound on the interaction strength 
was derived that is related to the Hamiltonian $H_S$. 
The relationship of the present work with it will be discussed elsewhere. 
\section{Spacetime uncertainty relation}\label{section6}
As an application of our results on the energy variance required for 
a measurement, we study the so-called spacetime uncertainty relation. 
It is believed that spacetime is not 
a simple continuum at the microscopic scale. 
A quantum effect is thought to impose some
limitation on the microscopic spacetime structure 
\cite{Mead}. 
This limitation has been proposed in various ways including 
string theory \cite{Yoneya}. 
In \cite{FDR}, Doplicher, Fredenhagen and Roberts 
gave an ingenious, yet heuristic argument 
by combining 
the quantum measurement of local observables with general relativity, 
as explained below. 
We may suppose that a spacetime region, say $D$, has its operational 
meaning if one is able to measure a local observable located at 
$D$. 
They argue that one needs to concentrate 
energy 
to measure an observable in $D$. 
If we assume that this energy can be identified with 
mass, for a sufficiently small $D$ the mass can become 
too large to avoid black hole formation. 
As a black hole prevents the extraction of information from  
the region, the region $D$ looses its operational meaning. 
(The authors further propose that the spacetime be described by 
noncommutative geometry.)  
Their reasoning with respect to the usage of 
general relativity is heuristic. In addition, 
the first part of this reasoning still contains an ambiguity. 
We apply our theorem to strengthen 
this first part of the energy-time uncertainty relation 
between apparatus energy fluctuation and measurement time. 
\par
We consider a nonrelativistic discrete space $\mathbb{Z}^d$. 
Thus the spacetime is $\mathbb{Z}^d \times \mathbb{R}$.  
There is no apriori reason that the whole part of the apparatus must 
be located in a small region specified by $D$. 
That is, one may consider a large apparatus that extends 
outside the small region to avoid energy concentration. 
Or, even if the apparatus is localized, it may interact with an infinitely extended 
environment that has large energy fluctuation. 
(We investigate the latter possibility below.) 
This loophole, however, is not applied if the locality 
of interaction is taken into consideration. 
Let us consider an apparatus located at the origin $0$ of the lattice 
(see FIG. \ref{figure_lattice}).
Other sites represent the environment. 
For each $x\in \mathbb{Z}^d$ there is a Hilbert space $\hi_x$ that is not 
necessarily finite dimensional. The background lattice structure $\mathbb{Z}^d$ 
introduces a natural Euclidean distance between sites denoted by 
$d(x, y)=\sqrt{\sum_i (x_i -y_i)^2}$.
For each finite region $\Lambda \subset \mathbb{Z}$, 
we consider a Hilbert space $\otimes_{x\in \Lambda} \hi_x$ and 
an observable algebra $\mathfrak{A}(\Lambda) \subset  
\mathbf{B}(\otimes_{x\in \Lambda} \hi_x)$. 
(See \cite{BR} for a precise 
definition of the mathematical terminology.)   
For $\Lambda_1 \subset \Lambda_2$, a natural inclusion relation 
$\mathfrak{A}(\Lambda_1) \subset \mathfrak{A}(\Lambda_2)$ is introduced.  
Thus an arbitrary pair of observables $A_j 
\in \mathfrak{A}(\Lambda_j)$ $(j=1,2)$
for any disjoint regions $\Lambda_1$ and $\Lambda_2$
satisfies 
$[A_1,A_2]=0$.  
The total observable algebra is defined by 
$\mathfrak{A}:= \overline{\cup_{\Lambda} \mathfrak{A}(\Lambda)}^{\Vert \cdot\Vert}$, 
where the summation is taken over all finite regions and 
it is made closed with respect to norm topology.  
The dynamics is specified by a local Hamiltonian. 
On each site a self Hamiltonian $h_x$ 
acting only on $\hi_x$ is defined. 
The interaction must be local. For simplicity, we 
assume that the interaction is nearest-neighbor. 
That is, only for each pair $\{x,y\}$ satisfying 
$d(x,y)=1$, an interaction 
$\Phi(\{x,y\})$ that acts on $\hi_x \otimes \hi_y$ is non-trivial.
For each finite region $\Lambda \subset \mathbb{Z}^d$, 
let $H_{\Lambda}$ denote its ``box Hamiltonian" defined by 
\begin{eqnarray*}
H_{\Lambda} = \sum_{x \in \Lambda} h_x+ \sum_{\{x,y\} \subset \Lambda} 
\Phi(x, y). 
\end{eqnarray*}
This box Hamiltonian defines the dynamics (in the Heisenberg picture) by 
$\beta^{\Lambda}_t(\A) = e^{i \frac{H_{\Lambda}}{\hbar}t}
\A e^{-i \frac{H_{\Lambda}}{\hbar}t}$, 
where $\beta^{\Lambda}_t$ is a one-parameter $*$-automorphism group 
$(t\in \mathbb{R})$. 
The infinite volume limit on the box converges and defines 
the dynamics, $\beta_t$, by 
\begin{eqnarray*}
\beta_t(\A) = \lim_{\Lambda \to \mathbb{Z}^d} \beta^{\Lambda}_t(\A), 
\end{eqnarray*}   
where the limit is taken with respect to norm topology. 
Let us consider a system to be measured. 
We assume that it interacts only with an apparatus and has
a Hilbert space $\hi$. 
Its own dynamics is governed by a system Hamiltonian $H_S$, 
which defines a one parameter group of $*$-automorphism 
$\gamma_t$ by $\gamma_t(\B) = e^{i \frac{H_S}{\hbar}t}
\B e^{i \frac{H_S}{\hbar}t}$ for $\B \in \mathbb{B}(\hi)$. 
The interaction is given by an interaction Hamiltonian $V$ acting only on 
$\hi \otimes \hi_0$. 
The full dynamics is given by the norm limit of the 
``box Hamiltonian" approximated dynamics 
$\alpha^{\Lambda}_t(\A) = e^{i \frac{H_S + H_{\Lambda}+V}{\hbar}t}
\A e^{-i \frac{H_S+ H_{\Lambda} + V}{\hbar}t}$ 
as
\begin{eqnarray*}
\alpha_t(\A) = \lim_{\Lambda \to \mathbb{Z}^d}
\alpha_t^{\Lambda}(\A). 
\end{eqnarray*}
\par
Let us reformulate Condition \ref{cond1} 
to treat this infinite quantum system. 
The total observable algebra is written as 
$\mathbf{B}(\hi) \otimes \mathfrak{A}$.   
We define $\alpha^0_t$ by 
$\alpha^0_t = \gamma_t \otimes \beta_t$ which describes 
the dynamics without interaction between the system and the apparatus-plus-environment. 
Let us denote by $\omega$ an initial state of the 
apparatus-plus-environment at time $t=0$. 
(For simplicity we set $t_0=0$.) 
An initial state of the composite system is written as 
$\rho \otimes \omega$ with a state $\rho$ of the system. 
Its real time evolution (in the Schr\"odinger picture) is represented as
$(\rho \otimes \omega)\circ \alpha_t$. 
Condition \ref{cond1} is now replaced by the following condition: 
\begin{figure}
\includegraphics[width=8cm]{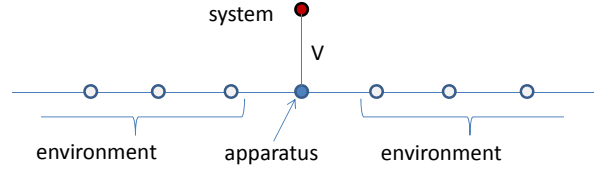}
\caption{The origin of $\mathbb{Z}^d$ represents an apparatus 
which interacts with a system and environment.}
\label{figure_lattice}
\end{figure}
\begin{condition}\label{cond4}(No-interaction up to switching-on time) 
For $t\leq 0$ and an arbitrary state $\rho$ of the system, $(\rho \otimes \omega) \circ \alpha_t 
= (\rho \otimes \omega) \circ \alpha^0_t$ holds. 
\end{condition}
For $\tau>0$, 
let us consider the states $\Theta_0$ and $\Theta_{\tau}$ defined by 
$\Theta_0 = \rho \otimes \omega$ and 
$\Theta_{\tau} = \rho \otimes ( \omega \circ \beta_{-\tau})$.   
Due to Condition \ref{cond4}, the latter state can be rewritten as
\begin{eqnarray*}
\Theta_{\tau} =
( (\rho \circ \gamma_{\tau}) \otimes \omega ) \circ 
\alpha^0_{-\tau}
= ( (\rho \circ \gamma_{\tau}) \otimes \omega ) \circ 
\alpha_{-\tau}. 
\end{eqnarray*}
We consider states at time $t=\tau$ 
for initial states $\Theta_0$ and $\Theta_{\tau}$.
Although $\Theta_0 \circ \alpha_{\tau}$ may 
become complicated, $\Theta_{\tau} \circ \alpha_{\tau}$ 
is written in a simple form as 
\begin{eqnarray*}
\Theta_{\tau} \circ \alpha_{\tau} = 
(\rho \circ \gamma_{\tau}) \otimes \omega. 
\end{eqnarray*}  
Suppose that the dynamics $\alpha_{t}$ describes 
a measurement process for a PVM $\P=\{\P_n\} 
\subset \mathbf{B}(\hi)$
with its measurement time duration $\tau$.
Let us take the Gelfand-Naimark-Segal (GNS) representation \cite{BR} of 
the apparatus plus environment. 
That is, a state $\omega$ is represented 
by using a Hilbert space $\hik$, a normalized 
vector $|\Omega\rangle \in \hik$ and 
a representation $\pi$ 
as
$\omega(A) = \langle \Omega | \pi_{\omega} (A) \Omega\rangle$.  
Let us consider the normalized vectors $|\phi_n\rangle \in \hi$ satisfying 
$\P_n|\phi_n\rangle = |\phi_n\rangle$. Each $|\phi_n\rangle$ defines 
a state $\rho_n(B) = \langle \phi_n |B|\phi_n\rangle$ 
on the system for $B \in \mathbf{B}(\hi)$. 
As the dynamics describes a measurement process of 
$\P$, there exists a POVM $\E= \{\E_n\} \subset
\mathbf{B}(\hi_0) \subset \mathfrak{A}$ 
such that 
\begin{eqnarray*}
(\rho_n \otimes \omega)\circ \alpha_{\tau}(\id \otimes \E_m) 
= \langle \phi_n\otimes \Omega| \pi(
\alpha_{\tau}(\id \otimes \E_m))|\phi_n \otimes \Omega\rangle=
\delta_{nm}, 
\end{eqnarray*}
where $\pi= id \otimes \pi_{\omega}$ is a representation 
of the total observable algebra. 
Let us consider a pair of states $\{\rho_{+}=|+\rangle \langle +|, 
\rho_-=|-\rangle \langle -|\}$ on the system defined by 
$|\pm\rangle = \frac{1}{\sqrt{2}}(|\phi_0 \rangle \pm |\phi_1\rangle)$.
For any $\A \in \mathbf{B}(\hi)$, it holds that 
\begin{eqnarray*}
|(\rho_{+}\otimes \omega - \rho_{-}\otimes \omega)
\circ \alpha_{\tau} ( \A \otimes \id) |
&=&2 | Im \langle \phi_0 \otimes \Omega | 
\pi(\alpha_{\tau}(\A \otimes \id))|\phi_1 \otimes \Omega\rangle | 
\\
&\leq& 2 | \langle \phi_0 \otimes \Omega| 
\pi(\alpha_{\tau}(\A \otimes \id) )|\phi_1\otimes \Omega\rangle | 
\\
&\leq& 2 \sum_n |\langle \phi_0 \otimes \Omega |
\pi(\alpha_{\tau}(\A \otimes \E_n))|\phi_1 \otimes \Omega \rangle |
\\
&= & 2 \sum_n |
\langle \phi_0 \otimes \Omega | 
\pi(\alpha_{\tau}(\id \otimes \E^{1/2}_n))
\pi(\alpha_{\tau}(\A \otimes \id))
\pi(\alpha_{\tau}(\id \otimes \E^{1/2}_n))|\phi_1 \otimes \Omega\rangle |
\\
&\leq&
2 \Vert \A \Vert 
\sum_n \langle \phi_0 \otimes \Omega| 
\pi(\alpha_{\tau}(\id \otimes \E_n))|\phi_0 \otimes \Omega\rangle^{1/2}
\langle \phi_1 \otimes \Omega | \pi(\alpha_{\tau}
(\id \otimes \E_n))|\phi_1 \otimes \Omega\rangle^{1/2}
\\
&=&
2 \Vert \A\Vert \sum_n 
\delta_{0n} \delta_{1n} =0.
\end{eqnarray*}
Let $\rho_{\pm}(\tau)$ denote the restricted states 
of $(\rho_{\pm} \otimes \omega) \circ \alpha_{\tau}$ on 
the system.  
We can conclude that 
the pair of states $\rho_+(\tau)$ and $\rho_-(\tau)$ 
satisfies $\rho_+(\tau) = \rho_-(\tau)$. 
Therefore we can introduce $\rho(\tau)$ 
by $\rho(\tau) = \rho_{+}(\tau) = \rho_{-}(\tau)$. 

We define a pair of states $|\pm'\rangle \langle \pm'|$ 
by $|\pm'\rangle =e^{-\frac{H_S \tau}{\hbar}} |\pm\rangle$. 
The argument in the previous section is applied to derive 
$\min\{F(\rho(\tau), |+'\rangle \langle +'|), 
F(\rho(\tau), |-'\rangle \langle -'|)\} \leq \frac{1}{\sqrt{2}}$.
Hereafter we 
assume $F(\rho(\tau), |+'\rangle \langle +'|) \leq \frac{1}{\sqrt{2}}$. 

In the present case, since the environment is infinite, 
the energy fluctuation of the environment is 
usually infinite as well. Thus the direct application of 
our previous result gives only a trivial inequality. 
Considering the locality of the model, we can see that  
the region relevant to the dynamics of the system 
is finite, because the information propagates essentially at 
a finite speed. 
This finite information propagation speed $V_{\Phi}<\infty$ is 
given by the so-called Lieb-Robinson bound 
\cite{BR,NS}. 
This bound allows us to approximate the dynamics 
by that given by a ``box Hamiltonian". 
For an arbitrary $\epsilon >0$, and $v\geq V_{\Phi}$, there is a finite
region $D$ such that for any $\A \in \mathbf{B}(\hi)$ and $t\in \mathbb{R}$,  
\begin{eqnarray*}
\Vert \alpha_t(A) - 
\alpha_t^{\Lambda_{v|t|, D}}(A) \Vert 
\leq \epsilon \Vert \A\Vert,  
\end{eqnarray*}
where 
$\Lambda_{v|t|, D} := 
\{x |\ ^{\exists} y \in D,  d(x, y)<v|t|\}$ 
represents a ``box".   
Let us consider the approximated dynamics given by the box 
$\Lambda_{V_{\Phi} \tau, D}$. 
The states $\rho_{\pm}$ evolve
into $\rho'_{\pm}(\tau)$ in time $\tau$ 
following this approximated dynamics.
For the box Hamiltonian 
$H_{\Lambda_{V_{\Phi}\tau, D}}$, 
we apply Theorem \ref{ETtheorem}
to obtain 
\begin{eqnarray*}
\cos\left(\frac{\Delta H_{\Lambda_{V_{\Phi}\tau, D}}}{\hbar}\tau
\right)\leq
F(\rho'_{+}(\tau), |+'\rangle \langle +'|), 
\end{eqnarray*} 
which is equivalent to 
\begin{eqnarray}
\frac{\Delta H_{\Lambda_{V_{\Phi}\tau, D}}}{\hbar} \tau 
\geq \mbox{Arccos} F(\rho'_{+}(\tau), |+'\rangle \langle +'|).
\label{Arccoseq}
\end{eqnarray}
To estimate how well the state $\rho'_+(\tau)$   
approximates the real state 
$\rho_+(\tau)$, we use a result obtained by Rastegin 
\cite{Rastegin}. 
\begin{lemma}\label{lemmaRastegin}
For states $\rho_0, \rho_1, \sigma$, it holds that 
\begin{eqnarray*}
\mbox{Arccos} F(\rho_0, \rho_1) \leq \mbox{Arccos}
F(\rho_0, \sigma) + \mbox{Arccos} F(\rho_1, \sigma). 
\end{eqnarray*}
\end{lemma} 
Applying this to (\ref{Arccoseq}), we obtain 
\begin{eqnarray*}
\frac{\Delta H_{\Lambda_{V_{\Phi}\tau, D}}}{\hbar} \tau 
\geq \mbox{Arccos}F(\rho_+(\tau), |+\rangle \langle +|)
-\mbox{Arccos}F(\rho_+(\tau), \rho'_+(\tau)).
\end{eqnarray*}
The last term on the right-hand side is bounded by the use of  
a relation between fidelity and trace distance,  
$1-F(\rho_+(\tau), \rho'_+(\tau))\leq D(\rho_+(\tau), 
\rho'_+(\tau))
:= \sup_{\A: \Vert \A\Vert=1} 
|\mbox{tr}[(\rho_+(\tau) -\rho'_+(\tau))\A]|
$. Since it holds that  
\begin{eqnarray*}
1- F(\rho_+(\tau), \rho'_+(\tau))^2
= (1-F(\rho_+(\tau), \rho'_+(\tau)
(1+F(\rho_+(\tau), \rho'_+(\tau))) 
\leq 2 D(\rho_+(\tau), \rho'_+(\tau)), 
\end{eqnarray*}
we obtain 
\begin{eqnarray*}
\mbox{Arccos}F(\rho_+(\tau), \rho'_+(\tau))
\leq \mbox{Arcsin} \sqrt{ 2 D(\rho_+(\tau), \rho'_+(\tau))}\leq \frac{\pi}{2} \sqrt{2D(\rho_+(\tau), \rho'_+(\tau))}.
\end{eqnarray*}
As $F(\rho_{+}(\tau), |+'\rangle \langle +'|)\leq \frac{1}{\sqrt{2}}$ holds, 
we conclude 
%
\begin{eqnarray*}
\Delta H_{\Lambda_{V_{\Phi}\tau, D}} \cdot \tau \geq 
\frac{\pi \hbar}{4}- \frac{\pi \hbar}{2}\sqrt{2 \epsilon},  
\end{eqnarray*}
which represents a starting point in  
the discussion of the spacetime uncertainty relation. 
\par
Let us give a rough sketch on the heuristic derivation of the 
spacetime uncertainty relation following \cite{FDR}. 
For the relativistic dynamics, $V_{\Phi}$ 
in the above nonrelativistic model 
could be replaced by the speed of light $c$ 
and $D$ represents a relevant region of the local observable $\A$, and  
$\epsilon =0$. 
Thus we obtain for $\Gamma_D(\tau)= \{D+{\bf x} \in \mathbb{R}^3| 
{\bf x}^2 \leq (c\tau)^2\}$, 
\begin{eqnarray*}
\Delta H_{\Gamma_D(\tau)} \tau \geq \frac{\pi \hbar}{4}. 
\end{eqnarray*}  
Assume that $D$ is a sphere with radius $R$. 
Then the volume of $\Gamma_D(\tau)$ 
is equal to 
$\frac{\pi}{3} (R+c\tau)^3$. 
Let us identify $\Delta H_{\Gamma_D(\tau)}$ with a mass $M c^2$ in 
the region $\Gamma_D(\tau)$. 
The above inequality gives 
\begin{eqnarray*}
M c^2 \tau \geq \frac{\pi \hbar}{4}.
\end{eqnarray*} 
For the mass to avoid the formation of a blackhole, 
$R + c\tau$ must exceed the Schwarzschild radius 
and 
\begin{eqnarray*}
R+c\tau \geq \frac{2 GM}{c^2}
\end{eqnarray*}
holds. 
Thus we have the limitation, 
\begin{eqnarray*}
\tau (R+c\tau) \frac{c^4}{2G}\geq \frac{\pi\hbar}{4}.   
\end{eqnarray*}
For small $\tau (\ll \frac{R}{c})$, this implies
\begin{eqnarray*}
\tau R \geq \frac{\pi \hbar G}{2 c^4}. 
\end{eqnarray*}
\section{Discussions}\label{section7}
In this paper, we studied quantum measurements as physical processes.
Each measurement is described as an interaction between 
a system and an apparatus. We investigated the energy fluctuation of the apparatus 
and the strength of interaction so that the system-plus-apparatus 
is regarded as a closed quantum system. 
We first examined the so-called standard model 
of measurement to find that this model needs an external system 
that switches on its measurement interaction. In this sense, the model is 
not closed.   
For the system-plus-apparatus 
to be genuinely closed, the quantum side must be made large enough 
to be capable of switching on the interaction autonomously. 
In this setting, we showed a trade-off relation between 
the measurement time duration and the energy fluctuation of the apparatus. 
This relation was obtained by combining two uncertainty relations; 
the information-disturbance relation and Mandelstam-Tamm uncertainty 
relation. 
We applied this relation to strengthen an argument regarding the spacetime 
uncertainty relation. 
In addition, we considered the strength of interaction between 
the system and the apparatus, and  
showed that there also exists a trade-off between 
the strength and the measurement time duration.
\par
For an apparatus to evolve completely freely up to a 
certain time and then 
switch on a perturbation, the Hamiltonian must be two-sided
unbounded. This unbounded nature is often unwelcome  
as it 
forces the apparatus to be an infinite or unstable system. 
To avoid this situation, we may weaken Condition \ref{cond1} by allowing some error
(see Example \ref{example:real}). This kind of weakening of the conditions has been often employed 
in studies on quantum clocks \cite{Peres, BDM}. 
A detailed study on more ``physical" measurement models will be treated elsewhere. 
\par
A quantum apparatus in our problem can be understood as a device performing 
a specified task. This sort of device is studied in a context of 
programmable quantum gates 
\cite{HZbook, Nielsen}
and the optimality of control 
has been recently discussed 
\cite{Bisio}. 
It would be interesting to 
extend our result so as to be applicable to various tasks other than 
measurement. 
For instance the famous Einstein photon box 
\cite{BuschTEUR, Busch1} could, we hope, be revisited.  
\par
While our treatment of spacetime uncertainty relations agrees with 
physical intuition, there remain things to be improved. 
The discrete space 
is the most crucial drawback in our model. We should treat 
the spacetime as a continuum, or as a 
flat Minkowski spacetime as the first 
natural setting. In order to achieve this, the time translation group 
$\mathbb{R}$ should be generalized to 
the Poincar\'{e} group $\mathcal{P}_+^{\uparrow}$.
Then the localization of event will be made much clearer and 
a covariant spacetime uncertainty relation is expected to be obtained. 
Such an extension and refinement of our analysis to $\mathcal{P}_+^{\uparrow}$ 
(or general Lie group $G$) is an important problem which we hope will be
studied elsewhere. 
\par
As a final remark, we want to point out that an external classical clock is 
assumed in our study. If we shift the interface so that the quantum side 
contains a clock, the time is not anymore 
an extrinsic parameter but instead will become 
a relative quantity depending on the quantum clock \cite{PageWootters, 
Milburn, Giovannetti}. 
A related concept, the notion of reference frame, has been recently 
studied extensively 
\cite{BRS, mlbshort,VSP}. 
%
\vspace{2mm}
\\
{\bf Acknowledgements}
I am grateful to anonymous referees for valuable comments, 
and to Leon Loveridge for many helpful remarks.  
This work was supported by KAKENHI Grant Number 15K04998. 
\appendix
\section{Proof of Theorem \ref{theoremappendix}
}\label{AppendixA}
\begin{proof}
For simplicity, we assume that $t_0=0$ and that 
$\sigma(0)(=|\phi(0)\rangle \langle \phi(0))|$ is pure. 
Let us consider two states $\Theta_0(0)$ and $\Theta_t(0)$ $(0\leq t\leq \frac{\pi \hbar}{2 \Delta H_c})$ defined by 
\begin{eqnarray*}
\Theta_0(0)&:=&\rho(0)\otimes |\phi(0)\rangle \langle \phi(0)| 
\\
\Gamma_t(0)&:=&\rho(0) \otimes |\phi(-t)\rangle \langle \phi(-t)|. 
\end{eqnarray*}
These states evolve with the Hamiltonian $H=H_S+H_A+V$. 
Let us denote the states at time $t$ by $\Theta_0(t)$ and $\Gamma_t(t)$. 
While $\Theta_0(t)$ may have a complicated form, $\Gamma_t(t)$ has 
a simple form,
\begin{eqnarray*}
\Gamma_t(t)=\rho^0(t)\otimes |\phi(0)\rangle \langle \phi(0)|, 
\end{eqnarray*}
where we used 
Lemma \ref{lemma1}.
Because the fidelity between two states is invariant 
under unitary evolution \cite{Nielsen}, it follows that 
\begin{eqnarray*}
F(\Theta_0(0), \Gamma_t(0))
=F(\Theta_0(t), \Gamma_t(t)). 
\end{eqnarray*}
The left-hand side of the above equation becomes 
\begin{eqnarray*}
F(\Theta_0(0), \Gamma_t(0))
=|\langle \phi(0)| \phi (-t)\rangle |
\end{eqnarray*}
and the right-hand side is bounded as 
\begin{eqnarray*}
F(\Theta_0(t), \Gamma_t(t))\leq F(\rho(t), \rho^0(t)), 
\end{eqnarray*}
where we utilized the fact that the fidelity decreases for restricted states \cite{Nielsen}. 
Thus it holds that 
\begin{eqnarray*}
|\langle \phi(0)|\phi(-t)\rangle |\leq F(\rho(t),\rho^0(t)). 
\end{eqnarray*} 
The left-hand side of this inequality represents the speed of 
time evolution of the apparatus and is bounded. 
Let us fix a value $0 \leq F_0 \leq 1$ and 
denote by $\tau_F$ the minimum time $t$ 
attaining $F(\rho(t), \rho^0(t)) \leq F_0$. 
Then we obtain,   
\begin{eqnarray*}
\tau_{F_0} \cdot \Delta_{\alpha}(H) 
\geq 2 \hbar
\arccos
\frac{F_0 + 1-\alpha}{\alpha}. 
\end{eqnarray*} 
For this process to describe 
a measurement process, there must be an initial state 
attaining $F(\rho(t), \rho^(0))\leq \frac{1}{\sqrt{2}}$. 
Thus we obtain 
\begin{eqnarray*}
\tau \cdot \Delta_{\alpha}(H) 
\geq 2 \hbar \arccos \frac{\frac{1}{\sqrt{2}}+ 1 -\alpha}{\alpha}.
\end{eqnarray*}
\end{proof}
\begin{corollary}
For $\alpha = \frac{1}{2}\left(\frac{1}{\sqrt{2}}+1\right)$, 
it holds that 
\begin{eqnarray*}
\tau \cdot \Delta_{\alpha}(H) 
\geq \frac{2 \pi \hbar}{3}. 
\end{eqnarray*}
\end{corollary}
\section{Proof of Theorem \ref{th:error}}\label{sec:error}
\begin{proof}
We mimic the proof of Theorem \ref{th:mainmain}. 
We consider the dynamics  
from time $t=t_0$ to $t=t_0+\tau$. 
Suppose that $0$ and $1$ are possible outcomes. 
We consider two states $|0\rangle $ and $|1\rangle$ 
satisfying $P_0 |0\rangle = |0\rangle$, $P_1|0\rangle  = 0$, 
$P_1|1\rangle = |1\rangle$, and $P_0|1\rangle =0$ 
and define $|\pm\rangle := \frac{1}{\sqrt{2}}(|0\rangle + |1\rangle )$. 
We consider 
a pair of initial states $|\pm\rangle \langle \pm | \otimes \sigma(t_0)$.  
The states to be compared are 
$\rho_{\pm}:=\mbox{tr}_{\hik}[e^{-i \frac{H}{\hbar}\tau} 
(|\pm\rangle \langle \pm| \otimes \sigma(t_0))
e^{i \frac{H}{\hbar}\tau}]$. 
For an arbitrary operator $\A$ on the system, 
it holds that 
\begin{eqnarray*}
&&|\mbox{tr}[e^{-i \frac{H}{\hbar}\tau} (|(+\rangle \langle + |
- |-\rangle \langle - | )\otimes \sigma(t_0)) e^{i \frac{H}{\hbar}\tau}
(\A\otimes \id)]
\\
&\leq& 2 \Vert \A\Vert 
\sum_n \langle 0 \otimes \phi(t_0)| 
e^{i \frac{H}{\hbar}\tau} (\id \otimes \E_n)
e^{-i \frac{H}{\hbar}\tau}|0 \otimes \phi(t_0)\rangle^{1/2}
\langle 1 \otimes \phi(t_0)| 
e^{i \frac{H}{\hbar}\tau} (\id \otimes \E_n)
e^{-i \frac{H}{\hbar}\tau}|1 \otimes \phi(t_0)\rangle^{1/2}
\\
&=&
2 \Vert \A\Vert 
\sum_n P(n||0\rangle \langle 0|)^{1/2}
P(n | |1\rangle \langle 1|)^{1/2}
\\
&=&
2 \Vert \A\Vert 
\left( P(0||0\rangle\langle 0|)^{1/2} P(0||1\rangle \langle 1|)^{1/2}
+P(1||0\rangle \langle 0|)^{1/2} P(1||1\rangle \langle 1|)^{1/2} 
+ \sum_{n \neq 0,1} 
P(n||0\rangle \langle 0|)^{1/2} P(n||1\rangle \langle 1|)^{1/2} \right)
\\
&\leq &
2\Vert \A \Vert 
\left(P(0||1\rangle \langle 1|)^{1/2} 
+P(1||0\rangle \langle 0|)^{1/2} 
+(\sum_{n\neq 0,1} P(n||0\rangle \langle 0|))^{1/2}
(\sum_{n \neq 0,1}P(n||1\rangle \langle 1|))^{1/2}\right)
\\
&\leq &
6\Vert A\Vert \sqrt{P_{error}}.  
\end{eqnarray*}
Thus we obtain 
$D(\rho_+, \rho_-):= \sup_{\A: \Vert \A\Vert =1}
| \mbox{tr}[(\rho_+ -\rho_-)\A] | \leq 6 \sqrt{P_{error}}$. 
To estimate the magnitude of this perturbation we consider 
unitary evolution governed by the Hamiltonian $H_S$. 
In time $\tau$, this ``unperturbed" dynamics changes 
$|\pm\rangle $ to a pair of orthogonal 
states $|\pm'\rangle$ of the system. 
We then estimate $F(\rho_{+}, |+'\rangle \langle +'|)$ 
and $F(\rho_-, |-'\rangle \langle -'|)$. 
As $|\pm'\rangle$ are orthogonal, we have  
\begin{eqnarray*}
F(\rho_+, |+'\rangle \langle +'|)^2 + F(\rho_-, |-'\rangle \langle -'|)^2 
&=& \langle+'| \rho_+|+'\rangle + \langle -'| \rho_- |-'\rangle 
\\
&
=
& \langle+'|\rho_+ |+'\rangle + \langle-'|\rho_+|-'\rangle 
+\langle-'| \rho_- - \rho_+|-'\rangle
\\
&\leq &
 \mbox{tr}[\rho_+] + D(\rho_+, \rho_-) 
\\
&\leq& 1 + 6 \sqrt{P_{error}}.  
\end{eqnarray*} 
Thus we can conclude 
\begin{eqnarray*}
\min \{ F(\rho, |+'\rangle \langle +'|), F(\rho, |-'\rangle \langle -'|)\} 
\leq \sqrt{ \frac{1+6 \sqrt{P_{error}}}{2}}. 
\end{eqnarray*}
We assume $F(\rho, |+'\rangle \langle +'|)\leq \sqrt{
\frac{1+6 \sqrt{P_{error}}}{2}}$. 
Combining it with Theorem \ref{ETtheorem} 
by putting $\rho(t_0)=|+\rangle \langle +|$, 
we obtain 
\begin{eqnarray*}
\cos\left(
\frac{\tau \Delta H_A}{\hbar}  
\right)\leq \sqrt{ \frac{1+6 \sqrt{P_{error}}}{2}}.
\end{eqnarray*}
This ends the proof. 
\end{proof}


\begin{thebibliography}{9}
\bibitem{StateReduction}
While a measurement process consists of the ``information transfer" 
described by the time evolution of a composite system and 
the state reduction conditional to each measurement outcome, 
we treat only the former. See \cite{Ozawareduction} for the relevant 
discussions.
\bibitem{Hamiltonian}
For a given total Hamiltonian $H$ there is an ambiguity in    
dividing it into three parts as $H=H_S + H_A + V$. In this paper, 
by imposing conditions which are explained later, this arbitrariness is 
reduced to some extent. Although there remains arbitrariness 
even with these conditions, our results do not depend on the 
choice of the division.
\end{thebibliography}
\end{document}